\documentclass[11pt,a4paper,english]{article}
\usepackage[hmargin=2.5cm,vmargin=2.5cm]{geometry}
\usepackage{hyperref}
\usepackage{titlesec}
\usepackage{amsmath}
\usepackage{mathtools}
\usepackage{makeidx}
\usepackage{amsthm}  
\usepackage{amssymb} 
\usepackage{bbm}
\usepackage[nomath,notext]{kpfonts} 
\usepackage{mathrsfs}
\usepackage{enumitem}
\usepackage{bigints}
\usepackage{color}
\usepackage{xcolor}
\usepackage{tikz}  
\usepackage{float}

\usepackage{ifthen}
\usepackage{bm}
\usepackage{fix-cm} 
\usepackage{url}
\usepackage{marvosym}
\usepackage{caption}
\usepackage{subcaption}

\usepackage[normalem]{ulem}

\usepackage{authblk}

\setlist[enumerate,1]{label=\roman{*}), ref=\roman{*})}

\theoremstyle{plain}
\newtheorem{axiom}{Axiom}

\newtheorem{remark}[axiom]{Remark}
\newtheorem{theorem}{Theorem}[section]
\newtheorem{lemma}[theorem]{Lemma}
\newtheorem{proposition}[theorem]{Proposition}
\newtheorem{corollary}[theorem]{Corollary}
\newtheorem{definition}[theorem]{Definition}

\DeclareMathAlphabet{\mathsfbd}{T1}{\sfdefault}{\bfdefault}{\itdefault}
\SetMathAlphabet{\mathsfbd}{bold}{T1}{\sfdefault}{\bfdefault}{\itdefault}

\DeclareMathAlphabet{\mathsfbdit}{T1}{\sfdefault}{\bfdefault}{\itdefault}
\SetMathAlphabet{\mathsfbdit}{bold}{T1}{\sfdefault}{\bfdefault}{\itdefault}

\DeclareMathAlphabet{\mathsfit}{T1}{\sfdefault}{\mddefault}{\sldefault}
\SetMathAlphabet{\mathsfit}{bold}{T1}{\sfdefault}{\bfdefault}{\sldefault}

\definecolor{redwork}{rgb}{0.73,0.04,0.04}
\definecolor{applegreen}{rgb}{0.55, 0.71, 0.0}
\definecolor{ao}{rgb}{0.0, 0.5, 0.0}

\renewcommand{\(}{\left(}
\renewcommand{\)}{\right)}
\renewcommand{\[}{\left[}
\renewcommand{\]}{\right]}
\newcommand{\lbr}{\left\{} 
\newcommand{\rbr}{\right\}} 
\newcommand{\abs}[1]{\left\lvert #1 \right\rvert}

\DeclarePairedDelimiter\ceil{\lceil}{\rceil}

\newcommand{\rexp}{\mathcal{E}}

\newcommand{\eps}{\epsilon}

\renewcommand{\vec}[1]{\bm{#1}}

\newcommand{\N}{\ensuremath{\mathbb{N}}}  
  
\newcommand{\R}{\ensuremath{\mathbb{R}}}

\newcommand{\restr}[2]{{
  \left.\kern-\nulldelimiterspace 
  #1 
  \vphantom{\big|} 
  \right|_{#2} 
  }}

\newcommand{\Ind}[1]{
\ensuremath{\mathbbm{1}_{#1}}
} 

\renewcommand{\P}{\ensuremath{\mathbbm{P}}} 
\newcommand{\E}{\ensuremath{\mathbbm{E}}}  
\newcommand{\PPP}{\ensuremath{\mathcal{P}}}  

\DeclarePairedDelimiterX{\IP}[2]{\langle}{\rangle}{#1, #2} 

\newcommand{\dif}[1]{\hspace{3pt}d#1}
\newcommand{\bO}{\mathcal{O}} 
\newcommand{\lo}[1]{o\left(#1\right)} 

\newcommand{\norm}[1]{\ensuremath{\left\lVert #1 \right\rVert}} 

\newcommand{\Lp}[1]{\mathsf{L}^{#1}} 

\newcommand{\PS}[1]{\mathscr{P}_{#1}} 

\newcommand{\ze}{\mathbf{0}}

\renewcommand{\ge}{>}

\newcommand{\set}[1]{\mathcal{#1}}

\providecommand{\keywords}[1]
{
  \small	
  \textbf{\textit{Keywords---}} #1
}
\renewcommand{\PPP}{\Psi}

\begin{document}
\title{Genealogical transition in the noisy $N$-Branching Random Walk. How stronger selection may promote genetic diversity.}

\author[1]{Emmanuel Schertzer}
\author[2]{ Alejandro H. Wences}
\affil[1]{Faculty of Mathematics, University of Vienna}
\affil[2]{LAAS-CNRS, Université de Toulouse, CNRS}

\maketitle
\begin{abstract}
We consider an extension of the noisy $N$-Branching Random Walk that models the evolution of a population subject to natural selection. We show the existence of a critical value for the noise which separates the limiting genealogical structure into two regimes, which we respectively call the semi-pulled and the fully-pulled regimes. In the fully-pulled regime, the genealogy converges to a discrete time Poisson-Dirichlet coalescent. In the semi-pulled regime, the genealogy converges to the Bolthausen-Sznitman coalescent. We discuss some interesting  biological consequences of this result. In particular, our model predicts a non-monotone relation between the selection strength and the effective population size.
\end{abstract}
\keywords{Population genetics, Coalescent theory, Noisy $N$-Branching Random Walk, Pulled fitness waves and their genealogies}

\section{Introduction.}

Darwinian evolution emerges from the processes of mutation, reproduction, competition, and genetic drift.
To model these evolutionary forces, an {\it asexual} population can be naturally encoded as a {\it fitness wave}: a cloud of individuals moving in an abstract $1$-dimensional fitness space \cite{BrunetDerridaMullerMunier2007,RouzineWakeleyCoffin2003, RouzineBrunetWilke2008, NeherHallatschek2013}. 
Assuming discrete-time dynamics, each generation is generally subdivided into two sub-phases: a {\it reproduction phase} followed by a {\it ``culling'' phase}, which constraints the population size to remain at equilibrium. See Figure \ref{fig:model}. 
In the reproduction phase,  each individual gives birth to a random number of individuals whose fitness deviates from the parental fitness due to the accumulation of random mutations along the genome \cite{NeherHallatschek2013}. The reproduction phase generally increases the population size whereas the culling phase, by selecting individuals according to their fitness, maintains the population at  a fixed size $N$.  

Fitness wave models have attracted quite a lot of attention in the Physics and Mathematics literature in the last two decades, see e.g., 
\cite{BrunetDerridaMullerMunier2007,RouzineWakeleyCoffin2003, RouzineBrunetWilke2008, NeherHallatschek2013,DesaiFisher2007,DesaiWalczakFisher2013,Fisher2013,TsimiringLevineKessler1996,LiuSchweinsberg2023,RobertsSchweinsberg2020,Schweinsberg2017-I,Schweinsberg2017-II}. 
These models are natural extensions of the exchangeable (or neutral) paradigm for population models \cite{Cannings74,Cannings75,WencesJegousseJOMB2024}. 
However, even in the simplest settings,
they are notoriously difficult to analyse and, to this date, a large gap remains between our understanding of neutral models and more realistic models with natural selection. 

The aim of the present article is to uncover a new phase transition for fitness waves. We present an exactly solvable model where the parameter space is partitioned into two main regimes, that we call the {\it semi-pulled} and {\it fully-pulled} regimes,
{each characterized by its own genealogical structure}. As we shall see, this phase transition {in the shape of the genealogy } will have important implications for  the following fundamental question in population genetics: 
\begin{enumerate}
\item[(Q)] How is genetic diversity related to selection strength? In other words, if we increase selection strength, does it lead to a decrease or increase in genetic diversity within the population, measured by the expected number of neutral genetic differences between two randomly sampled individuals?
\end{enumerate}
Many results since the foundation of population genetics \cite{Fisher1923, Haldane1927, HaldaneJayakar1963, SmithHaigh1974}, point  to natural selection as  one of the main causes of reduction of genetic diversity in a population. Here we present a directional selection model with the unexpected feature  that stronger selection may induce a higher genetic diversity.
\par

\subsection{The model.}
\label{sect:pop}
We consider a population evolution model of $N$ particles on the real line which reproduce at discrete times. 
Every generation is of size $N$ and particles are  given 
a fitness value on the real line. For $j\leq N$, we denote by $X^N_j(t)$ the position of the $j^{th}$ particle at generation $t$, assuming that particles are ranked 
in decreasing order.
Particles at generation $0$ are placed according to any arbitrary configuration. 
At a given time $t\geq 1$, the population is constructed from
the previous generation in three consecutive steps: one reproduction step, a truncation selection step, and a
 natural selection step. \par

\begin{enumerate}{}{\labelwidth=12pt}
\item[] \emph{(Reproduction)} The children of individual $X^{N}_k(t-1)$ are given by the random point process 
\begin{equation}\label{eq:delta-k}
 \Delta_{k}(t-1)
\end{equation}
where conditionally on $(X_k(t-1);k\in[N])$,  $\Delta_{k}(t-1)$
is an independent Poisson Point Process (PPP) of intensity measure 
$e^{-(s-X_k^N(t-1))}\dif{s}$. After the reproduction step, all parents die out. 
\item[] \emph{(Truncation Selection)} This is parametrized by  $\rho\geq 1$. 
In this step we only retain the $N^\rho$-rightmost individuals among those produced in the previous reproduction step. This is well defined since, almost surely, there are only finitely many particles on any right interval $[a,\infty)$, $a\in\R$.
.   
\item[] \emph{(Natural Selection)} This is parameterized by
$\beta>0$.
Conditionally on the positions of the surviving $N^\rho$ particles, we sample without replacement
$N$ individuals with the probability of picking an individual at position
$y$ being proportional to $e^{\beta y}$. 
\end{enumerate}

When $\rho=1$, the  natural selection step becomes innocuous, keeping always the $N$ right-most newly produced particles.  
After the reproduction and truncation selection steps, the effect of the natural selection step is to bring the total population size back to its equilibrium $N$ by selecting fit individuals 
through a noisy mechanism parameterized by $\beta$. The larger the $\beta$ is, the less noisy this procedure is, so that fittest (rightmost) particles do indeed
make it to the next generation. When $\beta=\infty$ the natural selection step is ``perfect'' in the sense that the $N$ right-most  individuals  are always chosen (thus coinciding with truncation selection of parameter $\rho=1$); whereas if $\beta=0$ then 
individuals are chosen uniformly at random, which corresponds to neutral evolution.
Thorough biological interpretations will be given in Section \ref{bio-interpretation}.\par
The above model
 interpolates between several models in the literature. The cases $\rho=\infty$ and $\beta=\infty$, and $\rho=1$ ($\beta\in\R$, the natural selection step is innocuous), are equivalent to the original integrable exponential model of 
 Brunet and Derrida \cite{BrunetDerridaMullerMunier2006,BrunetDerridaMullerMunier2007,BrunetDerrida2012}.
The case $\rho=\infty$ and $\beta\in(1,\infty)$ coincides with the model introduced by Cortines and Mallein in \cite{CortinesMallein2017}, who find similar
phenomenology as in the Brunet and Derrida model. In this setting, they also proved that the natural selection step is not well defined whenever $\rho=\infty$ and $\beta\leq 1$.
Here we avoid this impasse by introducing the truncation selection step parameterized by $\rho\geq 1$, and focus on the regime $\beta\in (0,1)$ for $\rho \in (1,\infty)$.

\begin{figure}[h]
\centering
\begin{tikzpicture}[x=1pt,y=1pt, scale = 1]
\input{ModelIntro.tex}
\end{tikzpicture}
\caption{Simulations of the passage from generation $t$ to generation $t+1$ with $N=10$, $\rho=2$, and $\beta=0.8$. Time flows from top to bottom. 
The first layer depicts the reproduction and truncation selection steps combined, resulting in a collection of
$N^\rho$ particles.
In the  natural selection step (second layer) 
$N$ particles are randomly selected according to their fitness to maintain the population size constant.\\
} 
\label{fig:model}
\end{figure}

{\bf Structure of the manuscript:} In Section \ref{sec:mainResults} we present our main mathematical results. In Section \ref{sec:discussion}, we  discuss some implications of our mathematical 
results to biology and the traveling wave literature.  Some heuristics of our main results are presented in Section \ref{sect:heuristics}. Complete proofs of the existence of the phase transition are given in Sections \ref{sec:proofsWeakSelection} and \ref{sec:proofsStrongSelectino}.
which rely on an integrability property of our exponential model presented in Section \ref{sect:poisson}. 
Finally, we postpone some technical but straightforward large deviations and uniform integrability results to the Appendix.

\section{Main results.}\label{sec:mainResults}

We investigate  the ancestral lines, the genealogy,  and the speed of evolution of the population in the case $\beta\in(0,1)$ and $\rho\in(1,\infty)$.  
Not surprisingly, stronger selection always increases the speed of the fitness wave (Theorem \ref{th:speedSelection} and Figure \ref{fig:SpeedT2VsBeta2}). However, the model exhibits a surprising phase transition 
when investigating the genealogical structure of the population (Theorem \ref{th:coaConvergence}). More specifically, we observe  a phase transition in the parameter space $(\rho,\beta)$. For a fixed value of $\rho$,
define the critical parameter
$$
\beta_c \equiv \beta_c(\rho):= 1-\frac{1}{\rho}.
$$
The regime $\beta\in (0,\beta_c)$ will be referred to as the {\it fully-pulled} regime, and the 
regime $\beta\in(\beta_c,1)$ as the {\it semi-pulled} regime. 
In section \ref{sec:discussion} below, we provide thorough phenomenological discussions around this phase transition and will justify the terminology in terms of two distinct propagation regimes for the fitness wave. \par

 Our first result provides a characterization of the speed of evolution.
We have the following: 
\begin{theorem}[Speed of evolution]\label{th:speedSelection}  
There exists a deterministic constant $\nu^N$  such that 

\begin{equation}\label{eq:speedSelectionAsmp}
\nu^{ N} = \lim_{t\to\infty}  \frac{X^N_1(t)}{t} = \lim_{t\to\infty} \frac{X^N_N(t)}{t} \ \ \mbox{a.s..}
\end{equation}
Further,  as $N\to\infty$,
\begin{description}
\item[i) Fully-pulled regime ($\beta<\beta_{c}(\rho)$):]
\begin{equation}
 \nu^N= - (\rho - \frac{1}{1-\beta})\log(N) + \E\[\log(Y_\beta)\] + \lo{1}
\end{equation}
where $Y_\beta$ is the positive $(1-\beta)$-stable law with Laplace transform 
$\E\[e^{-\lambda Y_\beta}\]=\exp\lbr-\rho\beta\lambda^{1-\beta} \rbr$.
\item[ii) Semi-pulled regime ($\beta>\beta_{c}(\rho))$:]
\begin{equation}
\nu^N=\log\(\chi\log N\) + \lo{1}.
\end{equation}
where 
\begin{equation}\label{eq:defChi}
\chi\equiv \chi(\beta,\rho)\coloneqq\frac{1-\rho(1-\beta)}{\beta}\in(0,1].
\end{equation}
\end{description}
\end{theorem}\par

\bigskip
Let $t\geq1$ and
let $\tilde{R}^N_1(t)$ be the rank of the parent in generation $t$ of a randomly
sampled individual in generation $t+1$, where particles in generation $t$ are ordered decreasingly according to their position on the real line (the fittest has rank $1$, and the less fit has rank $N$).\par 
Let 
$\PS{[0,1]}\coloneqq\ \cup_{k
\in \N\cup\{\infty\}}\{(\eta_1,\cdots,
\eta_{k})\colon \eta_1\geq \eta_2\geq \cdots; \sum_{i=1}^k \eta_i = 1\}$ be the space of (proper) mass partitions of the unit interval
(e.g. Definition 2.1 \cite{Bertoin2006});
and consider $\bm\eta^\infty=(\eta^\infty_1,\eta^\infty_2,\cdots)$ a random mass partition with Poisson-Dirichlet$(1-\beta,0)$ distribution \cite{PitmanYor97}.
{By Proposition 6 in \cite{PitmanYor97}, $\bm\eta^\infty$ can be constructed by setting $\eta^\infty_i=Z_i/\sum_{j=1}^\infty Z_j$,
where  $Z_1> Z_2>\cdots$ are the atoms of a PPP of intensity $(1-\beta) \frac{dx}{x^{\beta-2}}$ on $[0,\infty)$ arranged in decreasing order. Here we furthermore note that the random variable $Y_\beta$ appearing in Theorem \ref{th:speedSelection} above can be written as $Y_\beta=\sum_{j=1}^\infty Z_j$.}
\par
The next result will justify our terminology and will be discussed in more details in Section \ref{sect:semi-full}.
In the following $\Longrightarrow$ will mean convergence in distribution.
\begin{theorem}[Parental rank]\label{thm:typical}
As $N\to\infty$, for every $t\ge 1$, 
\begin{description}
\item[i) Fully-pulled regime ($\beta<\beta_{c}(\rho)$):] 
$\tilde{R}^N_1(t)\Longrightarrow R_1^\infty$  where 
$\P\(R_1^\infty=k\)= \E\[\eta^\infty_k\].$
\item[ii) Semi-pulled regime ($\beta>\beta_{c}(\rho))$:] 
$\log( \tilde{R}^N_1(t) ) /\log(N) \Longrightarrow \chi \cdot U$  where
$U$ is a uniform r.v. on $[0,1]$.
\end{description}
\end{theorem}\par
We now characterize the genealogy of the population.
For $n\leq N$, let $$
(\tilde\Pi^{N,T}_n(t) \equiv \tilde \Pi_n(t); 0\leq t \leq  T)
$$ be the process taking values in the space $\PS{n}$ of partitions of $[n]=\{1,\dots,n\}$ describing
the genealogy  
of $n$ randomly sampled individuals in generation $T\in\N$ of the population with size $N$. Formally, the random partition
 $\tilde \Pi^{N,T}_n(t)$ is the partition such that  two  indices are in the same block iff the corresponding sampled individuals share the same
 common ancestor at time $T-t$. Clearly $\tilde \Pi^{N,T}_n(0)=\{\{1\},\cdots,\{n\}\}$. \par

  In order to state the next result, 
  let us first recall the construction of discrete-time exchangeable coalescents through Kingman's paintbox coagulation procedure  \cite{Kingman1978}. The latter is driven by a random mass partition ${\bm\eta}\in \PS{[0,1]}$. Given a partition $\Pi\in \PS{n}$, we 
construct a coarser partition $\hat\Pi\in\PS{n}$ by merging (taking the union of) groups of blocks (elements) of $\Pi$. This is decomposed into two steps: (1) every
block of $\Pi$ is assigned an independent uniform random variable on $[0,1]$; (2) all those blocks 
whose uniform r.v. falls in the same interval of type $(\sum_{i=1}^{j}\eta_i, \sum_{i=1}^{j+1}\eta_i]$, with $j\geq1$, are merged together into a single block.

Discrete-time exchangeable coalescent processes
directed by a mass partition $\bm\eta$
consist of sequential Kingman's paintbox coagulations procedures each driven by an i.i.d. copy of $\bm\eta$. See e.g. Theorem 2.1 and Lemma 4.3 in Bertoin \cite{Bertoin2006}. We also refer the reader to \cite{Schweinsberg2000Xi,Bertoin2006} for a thorough exposition of the subject, and also to \cite{MohleSagitov2001, CortinesMallein2017, WencesJegousseJOMB2024} for general criteria for the weak convergence of discrete-time to continuous-time exhchangeable coalescent processes.

 Finally, let $\tilde T_2^N(T)$ be the time to the most recent common ancestor (TMRCA) of two sampled individuals in generation $T$ (setting $\tilde T_2^N(T)=\infty$ if they have different
 ancestors at generation $0$); and consider $\tilde R_2^N(T)$ the 
  rank of the MRCA at that time (setting $\tilde R_2^N(t)=\infty$ on the event $\tilde T^N_2(T) = \infty$).

\begin{definition}\label{definition-wi}
     Let $(\rexp_i)_{i\geq 1}$ be i.d.d. standard exponential random variables
and define 
$$w_i:= 1/(\rexp_1+\cdots+\rexp_i).$$ 

Conditionally on the vector $w_1>\dots>w_{N^\rho}$, sample without replacement $N$ indices $I^N \coloneqq \(I^N_1, \dots, I^N_N\)$ 
among $[N^\rho]$, assuming that the relative weight of $i$ is given by $w_i^\beta$.
Let $\hat I^N$ be the vector $I^N$ arranged in increasing order,  so that 
${\hat{\bm W}^N}\coloneqq ( w_{\hat I^N_1},\dots,w_{\hat I^N_N})$ is arranged in decreasing order. Define
    \begin{equation}\label{def:mass-partition}
{\bm \eta}^N \ := \frac{\hat{\bm W}^N}{\norm{\hat{\bm W}^N}_{\ell^1}} = \ \frac 1 { \sum_{k=1}^N w_{\hat I^N_k} } \  (w_{\hat I^N_1},\cdots, w_{\hat I^N_N} ).
\end{equation}

\end{definition}
  
\begin{proposition}\label{prop:ExpModelMultCoa}
For every
fixed $T\in\N$, we have 
$$
\(\tilde \Pi^{N,T}_n; 0\leq t \leq T\)  \ = \ \(\Pi^N_n(t); 0\leq t \leq T\) \ \ \mbox{in law} 
$$
where $(\Pi^N_n(t); t\geq0)$ is a discrete-time coalescent process directed by $\bm\eta^N$ and started at the partition of singletons $\{\{1\},\cdots,\{n\}\}\in\PS{n}$.\par 
Moreover, letting $R^N_2$ have distribution, conditionally on a realisation of $\bm\eta^N$, given by $\P(R^N_2=i)= (\eta^N_i)^2/\sum_{j=1}^N (\eta^N_j)^2$. Then,
conditionally on the event $\{\tilde T^N_2(T)\leq T-1\}$, we have
$$
\tilde R^N_2(t) = R^N_2 \text{ in law.}
$$
\end{proposition}
The main consequence of the previous result is that the exponential structure of the population model ensures that its genealogical structure  and the
ancestral rank of the MRCA can be characterized through $\Pi^N_n$ and $R_2^N$ respectively, where $\Pi^N_n$ is a an exchangeable coalescent. This corresponds to taking the horizon time $T\to\infty$ in our model.
As we shall see, the phase transition found in our model 
is a consequence of the limiting behavior of the mass partition $\vec \eta^N$ as $N\to\infty$.
The latter will depend on whether the quantity
\begin{equation}\label{eq:defAlpha}
 \alpha\equiv \alpha(\beta,\rho) \coloneqq \rho - \frac{1}{1-\beta}.
\end{equation}
is positive, corresponding to $\beta<\beta_c(\rho)$, or 
negative, which corresponds to $\beta>\beta_c(\rho)$.
\par
{Before presenting our result on the genealogy, let us first introduce the family of $n-\Lambda$-coalescents \cite{Pitman99,Sagitov99}
which are continuous-time Markov process  with values in $\PS{n}$. 
For a finite measure $\Lambda$ on $(0,1)$, the dynamics of the $n-\Lambda$-coalescent are the following: if the process is at state
$\pi=\{\pi_1,\dots,\pi_b\}$, a partition of $[n]$ with $b\leq n$ blocks, 
then any combination $\{\pi_{i_1},\cdots,\pi_{i_k}\}\subset\pi$  of $k$ blocks (with $2\leq k\leq b$) will merge into the single block $\cup_{j=1}^k \pi_{i_j}$ at rate
$$
\lambda_{b,k} = \int_0^1 x^{k-2}(1-x)^{b-k} \Lambda(dx).
$$
The new state is then always a coarser partition of $[n]$, until the absorbing
state $\{\{1,\cdots,n\}\}$ is reached. The Bolthausen-Sznitman coalescent \cite{BolthausenSznitman98} corresponds to the choice $\Lambda(dx)=dx$. This coalescent is of particular interest in Population Genetics  as it has been proposed
as a  null model for the genealogy of rapidly adapting populations \cite{NeherHallatschek2013}, motivating its widespread study (see e.g. \cite{DrmotaIksanovMohleRoesler2007,DrmotaIksanovMohleRoesler2009, GoldschmidtMartin2005, KerstingJegousseWences2021}). }
 
We have the following for the genealogy of our model:\par
\begin{theorem}[Genealogy]\label{th:coaConvergence}
Let $(\Pi^N_n(t); t\geq0)$ be as in Proposition \ref{prop:ExpModelMultCoa}. Then, as $N\to\infty$:
\begin{description}
\item [i) Fully-pulled regime ($\beta<\beta_{c}(\rho)$):]  
Let $\Pi^\infty_n$ 
 be a discrete-time coalescent with $n$ initial particles and directed by the random mass partition
 $\bm\eta^\infty$ of Theorem \ref{thm:typical} i). Then 
\begin{equation*}
  (\Pi^N_n(t); t\in\N) \Longrightarrow (\Pi_n^\infty(t); t\in\N)  
  \end{equation*}
for the product topology for $\(\PS{n}\)^\N$. In particular 
$\lim_{N\to\infty}\E\[T^N_2\] = \beta^{-1}.$\par
Furthermore, $R^N_2\Longrightarrow R^\infty_2$  where
\begin{align*}
\forall i\geq 1, \quad\P\(R^\infty_2 = i \) &=\frac{\E\[\(\eta^{\infty}_i\)^2\]}{\beta};
\end{align*}
\item[ii) Semi-pulled regime ($\beta>\beta_{c}(\rho))$:] 
The time-scaled process $$(\Pi^N_n(\lfloor {\chi \log(N)} t \rfloor); t\geq0)
$$ converges in distribution 
in the Skorohod topology $D([0,\infty), \PS{n})$ to the  $n$-Bolthausen-Sznitman coalescent. In particular,
$
\lim_{N\to\infty}\E\[T^N_2\]/\log N=\chi.
$\par
Furthermore, $R^N_2\Longrightarrow 1$ in distribution.

\end{description}
\end{theorem}

\section{Biological interpretation}\label{sec:discussion}

\subsection{Biological assumptions}
\label{bio-interpretation}

The reproduction phase in our model incorporates both a reproduction and a mutation mechanisms, where daughter particles have a mutated fitness (position on the real line) centered around the fitness of the parent. Our choice of the exponential point process in the reproduction phase  may seem artificial at first sight. Our motivation is two-fold. First, 
as in \cite{BrunetDerridaMullerMunier2006,BrunetDerridaMullerMunier2007,BrunetDerrida2012}, the model turns out to be exactly solvable. Second, extremal theory states that the extremes of a long sequence of Gaussian random variables converges (after scaling and centering)
to the same exponential point process (see e.g. Example 1.1.7 \cite{HaanFerreira2006}). Since fitness waves are usually driven by extremal particles \cite{BerestyckiBerestyckiSchweinsberg2013}, our model should be a good approximation of a model where individuals produce a large number of offspring with a Gaussian perturbation of their parental fitness. 
For the original model in \cite{BrunetDerridaMullerMunier2006,BrunetDerridaMullerMunier2007,BrunetDerrida2012}, more ``realistic'' models (which do not rely on the assumption of an exponential birth process) have been shown to exhibit an analog qualitative behavior, see e.g. \cite{BerestyckiBerestyckiSchweinsberg2013}.\par

The main objective of the truncation selection step is to fix the total carrying capacity of the environment, keeping the number of initial 
 offspring to be finite and fixed before the natural selection phase. Alternatively, this constraint could be enforced by assuming that every individual in the parental population has a fixed number of offspring. For instance, for every individual one can only retain the first $N^{\rho-1}$ rightmost individuals in the PPP encoding the position of its offspring. By doing so  the total number of offspring before the natural selection step is  $N^\rho = N^{\rho-1} \times N$ as in the present model. 
This model corresponds to the case where the number of offspring  produced during the reproductive phase is large ($r$-elected species \cite{MacArthurWilson2001}). Thus, our $\rho$ parameter can be re-interpreted as a measure of the number of offspring per individual.
  While this alternative model is presumably more realistic, it turns out to be much more challenging to solve than our model. However, the numerical simulations in Figure \ref{fig:T2fixedOffspSize}  show 
 that both models have similar qualitative behavior.

\begin{figure}[h]
\centering\includegraphics[scale = 0.4]{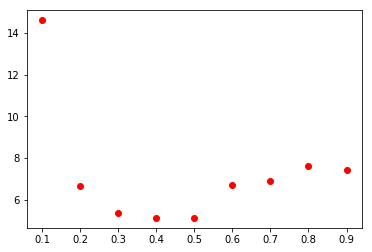}
\caption{ Model with no truncation selection but assuming that every individual produces $N^{\rho-1}$ children ($N=1000$ and $\rho=2$) during the branching step. Here the particles produced by a parent positioned at $x\in\R$ correspond
to the $N^{\rho-1}$ right-most particles of an independent PPP of intensity $e^{-(s-x)}ds$. As in our exactly solvable model, $\E(T_2^N)$ (y-axis) 
is non-monotone  as a function of $\beta$ (x-axis), i.e., genetic diversity  may increase as the strength of natural selection increases.}
\label{fig:T2fixedOffspSize}
\end{figure}

As mentioned above, the natural selection step is noisy and 
the level of noise is captured by the $\beta$ parameter of our model: $\beta=\infty$ consists in selecting the $N$ fittest individuals (no noise); whereas $\beta=0$ consists in selecting individuals uniformly at random. 
This feature of our model reflects that the survival of individuals depends not only on their
relative fitness but also on  random and/or confounding factors.

\subsection{Non linear response of evolution to natural selection}
{Our model predicts a complex response of evolutionary characteristics (speed of evolution and genetic diversity)
to the selection parameters $\rho$ and $\beta$. Let us fix the carrying capacity $N^\rho$ of the environment, so that the value of $\beta$ can be directly interpreted as a measure of the strength of natural selection. }

\bigskip

{\it Speed of evolution.} 
As expected, the speed of evolution is increasing with $\beta$ (Theorem \ref{th:speedSelection}); however, the relationship between speed and $\beta$ is highly non-linear (see right panel in Figure \ref{fig:SpeedT2VsBeta2}).
 As a matter of fact, as $N\to\infty$ and as $\beta$ increases, 
the gain on the speed of evolution becomes more and more negligible.

\bigskip

{\it {Genetic diversity}.} In population genetics, the effective population is often defined as 
$$
N_{e} := {\E[T_2^N]}.
$$
and measures the ``depth'' of the genealogy for two sampled individuals.
Assuming that neutral mutations accumulate at constant rate along the genome (as in the infinite sites model {\cite{Wright1931, GriffithsTavare1994}}) 
the expected number of neutral genetic differences (e.g. single nucleotide polymorphisms) between two individuals is proportional to the effective population size $N_e$ \cite{Durrett2008}. As a consequence, $N_e$ provides a measure of the expected neutral genetic diversity in the population.\par

{As mentioned above, many results  \cite{Fisher1923, Haldane1927, HaldaneJayakar1963, SmithHaigh1974}, point to natural selection as  one of the main causes of reduction of genetic diversity in a population.} 
{
Our model predicts that increasing the strength of selection may in fact induce a higher genetic diversity.
 We show that  $N_e$
decreases as a function of  $\beta$ on $(0,\beta_c)$ but, surprisingly,  it {\it increases} on the interval $(\beta_c,\infty)$ (see the left panel of Figure \ref{fig:SpeedT2VsBeta2} and
Theorem \ref{th:coaConvergence}). Indeed, even if we observe a
general reduction of $N_e$ when compared to neutral models --for which $N_e=\bO(N)$--, in our model it is on the order of $N_e=\bO(\log(N))$ in the ``weak selection regime'' ($\beta\in (0,\beta_c)$), and on the order of $\bO(1)$
in the ``strong selection regime'' ($\beta\in (\beta_c,1)$).}

\begin{figure}[h]
\centering
\scalebox{1}{\begin{tikzpicture}[x=1pt,y=1pt, scale = 0.9]
\input{SpeedT2VSBeta.tex}
\end{tikzpicture}}
\caption{The theoretical approximations for $\E\[T^N_2\]$ on the left, and the speed of evolution $\(\nu^N\)$ on the right, for $N=10^3$, along with simulated values (black circles), 
         as functions of the parameter 
         $\beta$ with fixed $\rho=2$ (thus $N^\rho=10^6$). 
  The phase transition occurs at $\beta_c=1/2$.  The genetic diversity is increasing with the selection strength in the strong selection regime. Further, for large values of $N$, the genetic diversity is typically higher in the strong selection regime as compared to the weak selection regime.}
  \label{fig:SpeedT2VsBeta2}
\end{figure}

\subsection{Semi and fully pulled waves.}\label{sect:semi-full}

Let us now describe the phenomenology of the phase transition observed in our model in terms of traveling wave dynamics.
According to Theorem \ref{thm:typical}, one common feature of the semi and fully pulled regime is that
$$
\tilde R_1^N(t) << N/2
$$
where we recall that $\tilde R_1^N(t)$
is the rank of the ancestor of a randomly sampled individual at time $t$.
In other words, at any given time, individuals always descend from a highly ranked parent, that is when compared to the average rank ($N/2$) of an individual chosen uniformly at random. 
Such propagation fronts are  referred to as {\it “pulled”} waves, because the wave is pulled along by the action of individuals located to the right of the front \cite{Saarloos2003}.   Theorem \ref{thm:typical} then uncovers a phase transition in the pulled regime which is new to our knowledge:
\begin{enumerate}
\item In the fully pulled regime, the wave is pulled by the very extreme individuals: the parental rank is $O(1)$ so that the wave at generation $t+1$ is generated by the individuals close to the right edge at generation $t$. 
\item In contrast, in the semi pulled regime, the wave is pulled by high ranked individuals, but those individuals are still far from the right edge: the parental rank is $\approx N^{\chi U}$. 
\end{enumerate}

\bigskip 

Semi- and fully- pulled fronts are characterized by very distinct genealogical structures. 
For semi-pulled waves,
the genealogy converges to the celebrated Bolthausen-Sznitman coalescent \cite{BolthausenSznitman98, Pitman99, Sagitov99}. 
This is consistent with previous findings in \cite{ BrunetDerridaMullerMunier2007, BrunetDerrida2012} where the  natural selection step consists in deterministically selecting the $N$-fittest individuals after the reproduction phase (truncation selection only).
The same universal behavior has been found in many population models (see \cite{Schweinsberg2017-II, CortinesMallein2017} for rigorous results, and 
\cite{NeherHallatschek2013, DesaiFisher2007, DesaiWalczakFisher2013} for non-rigorous arguments). 

One of the interesting features of the present model is that the Bolthausen-Sznitman universal behavior breaks down when the selection noise below the critical parameter $\beta_c$ and the transition from the semi to the fully pulled regime occurs. In the fully pulled regime, the genealogy then converges to a discrete time coalescent (Poisson-Dirichlet coalescent).

\bigskip

Let us now comment on the difference of time scale between the two regimes. 
In the fully pulled regime, ancestral lineages only visit sites at distance $O(1)$ from the front so that coalescence may occur with positive probability at every generation.
This implies that coalescence occurs on a time scale of order $1$ and the limiting genealogy is in discrete time.

In contrast,
the semi-pulled regime is driven by extreme events of the rightmost particle. 
 If we trace the ancestry of a sampled individual backwards in time, his lineage will typically be located far from the edge (Theorem \ref{thm:typical}ii)) where a coalescence with other lineages is unlikely to happen. However, at some exceptional past generation,  the right-most individual of the fitness wave happens to have a much higher relative fitness than the rest of the population. At that time, some of the ancestral lineages will jump to the rightmost point of the traveling wave ($R_2^N\Longrightarrow 1$ according to Theorem \ref{th:coaConvergence} ii)) and then coalesce.
 As indicated by Theorem \ref{th:coaConvergence}ii), the typical time scale to reach such generations is of order $\log(N)>>1$ and the limiting genealogy is now in continuous time.

\subsection{Comparison with other front propagation models}
Intricate phase transitions in front propagation have been observed in previous works.   
 One interesting example \cite{BirzuHallatschekKorolev2020, RoquesGarnierHamelKlein2012, Tourniaire2024, EtheridgePenington2022} is a population whose large scale behavior is described by the noisy F-KPP equation with Allee effect 
\begin{equation}\label{eq:B}
\partial_t u = \frac{1}{2} \partial_{xx} u + u(1-u)(1+Bu) + \sqrt{\frac{u(1-u)}{N}} \eta
\end{equation}
where $u\equiv u(t,x)$ is a  function describing the density of the population at time $t$ at the location $x\in\R$, $\eta$ is space-time white noise, $B>0$ encodes the magnitude of cooperation in the population, and $N$ is a large demographic parameter. \par
When $B>2$ the model exhibits {\it pushed} wave dynamics in the sense that the evolution and the genealogy are mainly driven by the {\it bulk} of the population. From a genealogical perspective, this implies that genealogical lines 
are located in the region the traveling wave where $u=O(1)$ (the bulk).

In turn, the pushed regime is itself partitioned into 2 sub-regimes: the semi-pushed regime ($B\in(2,4)$) and the fully-pushed regime ($B>4$).
In the fully-pushed regime, the fluctuations at the edge of the front obey the central limit theorem and the genealogy is conjectured {\cite{TourniaireSchertzer2024, BirzuHallatschekKorolev2020}} to converge to the Kingman coalescent \cite{Kingman1982} on a time scale of order $N$. In the the semi-pushed regime 
the fluctuations are described by a heavy tailed distribution and the genealogy  is conjectured \cite{BirzuHallatschekKorolev2018, FoutelSchertzerTourniaire2024, Tourniaire2024} to converge to a Beta coalescent \cite{Schweinsberg2003} on a time scale of order $N^a$
with $a<1$. 
In particular, the genealogical depth reduces when transitioning from the fully-pushed to semi-pushed regime.\par
When $B<2$, the population is now mainly driven by  individuals at the right edge of the front, and propagates according to what we now call a semi-pulled wave (previously referred to simply as a pulled wave). Analogously to our model in the strong selection regime, ancestral lines at coalescent times are typically placed away from the bulk population and the genealogy is conjectured {\cite{BirzuHallatschekKorolev2020}} to converge to the Bolthausen- Sznitman coalescent. 
\bigskip

Let us now compare the pushed regime for the noisy F-KPP equation with Allee effect described above with our results for pulled fitness waves. 
Our model is always pulled.
However, and analogously to (\ref{eq:B}),  the pulled regime is partitioned into two sub-regimes: the semi-pulled and fully-pulled regime.
Analogously to the pushed case, we observe that the genealogical depth also reduces when transitioning from the fully-pulled to the semi-pulled regime.
As a consequence, we observe a similar transition to (\ref{eq:B}), with the notable difference, that this transition 
is now driven by the particles at the edge of the front, as opposed to the particles in the bulk. 

\bigskip

The previous discussion 
is summarized in Table \ref{tab1}. 
In particular, we observe that even if all the models alluded are non-exchangeable (due to the presence of spacial structure or selection) the genealogies 
become exchangeable in the large population limit \cite{ Pitman99, Sagitov99}. This suggests that the universality class of exchangeable genealogies is presumably much larger than originally expected and that those models do not only emerge from discrete exchangeable models \cite{Schweinsberg2003} but also from from the complex interactions within travelling fronts.

\bigskip

\begin{table}[h]\label{tabel2}
\begin{center}
\begin{tabular}{|| c | c | c || c  | c ||}
 \hline
 & \mbox{Fully-Pushed} & \mbox{Semi-Pushed} & \mbox{Semi-Pulled} & \mbox{Fully-Pulled}  \\ [0.5ex] 
  & \mbox{$2<B<4$} & \mbox{$4<B$} &  $\beta > \beta_c$ & $\beta < \beta_c$  \\ 
 \hline\hline
 \mbox{TMRCA} &{$O(N)$} & { $O(N^a), a<1$} &  { $O(\log(N)^\alpha)$} & { $O(1)$}  \\ 
 \hline
 Genealogy &  { Kingman} & { $\mbox{Beta}(2-a,a)$}  & { Bolthausen-Sznitman} & { \mbox{PD$(1-\beta,0)$}}  \\ 
\hline

\end{tabular}
\caption{\label{tab1} Pushed transition for the noisy F-KPP equation (left two columns), and the pulled analog (right two columns) observed in our model.}
\end{center}
\end{table}

\section{Notation.}
We will write
\begin{align*}
f(x)\sim g(x) \text{ as }x\to a \text{ if }&\lim_{x\to a}\frac{f(x)}{g(x)} =1, \\
f(x)=\lo{g(x)} \text{ as }x\to a \text{ if }&\lim_{x\to a}\frac{f(x)}{g(x)} =0,
\end{align*}
and
$$
f(x)=\bO(g(x)) \text{ as }x\to a \text{ if }\lim_{x\to a}\frac{f(x)}{g(x)} <+\infty.
$$
We will also use the symbol ``$\Longrightarrow$'' to refer to convergence in distribution.\par
Also, in the following, $PPP(\mu)$ will denote a PPP of intensity $\mu$ on $\R$.

{ \section{Integrability}
\label{sect:poisson}
As highlighted in the studies by Brunet and Derrida \cite{BrunetDerrida2012} and Cortines and Mallein \cite{CortinesMallein2018}, the exponential structure of our population model offers an explicit representation of both the  speed of evolution and the genealogical structure of the population. This is the subject of the present section.
We first show that, after just one single generation, the relative positions of the particles reach an equilibrium distribution. 
\begin{proposition}\label{prop:equilibrium}
    For every $t\geq 1$, define $X^N_{eq}(t)\coloneqq \log(\sum_{i=1}^N e^{X^N_i(t)}).$ Then 
    $$
    \bigg(\left(X^N_1(t),\cdots, X^N_N(t)\right) -  X^N_{eq}(t-1); t\geq 1 \bigg)
    $$
    is a sequence of i.i.d. random variables distributed as follows. \begin{enumerate}
\item Take the $N^\rho$ right-most elements of a $PPP(e^{-x} dx)$.
\item Sample $N$ particles from the resulting set assuming that the sampling weight of a particle in $y$ is given by $e^{\beta y}$.
\end{enumerate}
\end{proposition}
\begin{proof}
Recall the definition of $\Delta_{k}(t)$ in (\ref{eq:delta-k}) which describes the relative position of the offspring of $X_k^N(t-1)$.
According to Proposition 1.3 in \cite{CortinesMallein2018},$$
\bigg(\sum_{k=1}^N \{\Delta_{k}(t) - X^N_{eq}(t-1)\}; t\geq1\bigg)
$$
is a sequence of i.i.d. PPPs with intensity measure $e^{-x} dx$. This can be easily seen from the Superposition Theorem for PPPs (see e.g. \cite{Kingman1992}). In our case: superimposing $N$ independent PPPs with intensity measure $e^{-(x-x_i)}dx$ yields a PPP with intensity measure $e^{-(x-x_{eq})}$ where $x_{eq}=\log(
\sum_{i=1}^N e^{x_i}
)$. Our Proposition then readily follows from the definition of our population model. 
\end{proof}
In order to prove Proposition  \ref{prop:ExpModelMultCoa}, we will need the following Lemma that follow from standard transformation identities (see e.g., the Mapping Theorem in \cite{Kingman1992}).
\begin{lemma}\label{le:pppsIdentities}
Let $\left(w_i = 
\frac{1}{{\cal E}_1+\cdots+{\cal E}_i}\right)_i$ as in Definition \ref{definition-wi}. Then
\begin{enumerate}
\item $(w_i)_i$ is identical in law to the ordered atoms of $\Psi$, where $\Psi$ is  $PPP(x^{-2}dx)$.
\item $(\log(w_i))_i$ is identical in law to the ordered atoms of $PPP\(e^{-x}dx\)$.
\end{enumerate}

\end{lemma}

\begin{proof}[Proof of Proposition \ref{prop:ExpModelMultCoa}]
Let $i\in[N]$ and $t\in\mathbb{N}$.
Define $A^N_i(t)$ to be the rank of the ancestor of individual $X_{i}^N(t)$. That is, $A^N_i(t)=j$ iff $X^N_j(t-1)$ is the parent of $X^N_i(t)$. Define
$$
{\cal H} := \sigma\{ X^N_j(t); j\in[N], t\in \mathbb{N}  \}.
$$
According to Lemma 1.6 in \cite{CortinesMallein2018},
the sequence $(A^N(t); t\geq1)$ is i.i.d. Moreover, it remains independent conditionally on ${\cal H}$ and 
for any ${\bf k}\in [N]^N$,
$$
\P\left( A^N(t+1) = {\bf k} \ | \ {\cal H} \right) \ = \prod_{i=1}^N \tilde \eta_{k_i}(t)
$$
where 
$$
\tilde \eta_{k}(t) \ = \ \frac{e^{X^N_{k}(t) -X_{eq}^N(t-1)}}{\sum_{j=1}^N e^{X^N_{j}(t) -X_{eq}^N(t-1)}}.
$$

We refer to the proof of Lemma 1.6. in \cite{CortinesMallein2018} for a short proof of the latter fact. Finally, the proof is completed by combining Proposition \ref{prop:equilibrium} and Lemma \ref{le:pppsIdentities}.
\end{proof}

Finally, the following analogue of Lemma 1.5. in   \cite{CortinesMallein2017}
is a first step towards the characterization of the speed of selection
in Theorem \ref{th:speedSelection}.
\begin{proposition}[]\label{prop:speed-of-selection}
The two limits 
$$\lim_{t\to\infty} \frac{X^N_1(t)}{t},\quad \lim_{t\to\infty}\frac{X^N_N}{t}$$ exist a.s. and
are equal to 
\begin{equation}\label{eq:speedSelection}
\nu^N = \E\[\log\(\sum_{k=1}^{N} w_{I^N_k}\)\].
\end{equation}
\end{proposition}
\begin{proof}
By Proposition \ref{prop:equilibrium}
the sequence 
$(X^N_{eq}(t+1)-X^N_{eq}(t); t\geq1$) is i.i.d. with finite mean. By the law of large numbers and the identities of Lemma \ref{le:pppsIdentities}, we have
\begin{equation}\label{eq:EMspeed1}
\lim_{t\to\infty}\frac{X^N_{eq}(t)}{t} \overset{a.s.}{=} \E\[\log\(\sum_{i=1}^N w_{I_i}\)\].
\end{equation}
The rest of the proposition  follows easily from
the fact that the sequences $(X_1^N(t) - X^N_{eq}(t-1); t\geq 1)$ and $(X_N^N(t) - X^N_{eq}(t-1); t\geq 1)$ are i.i.d. with finite mean.
\end{proof}}

{\section{Heuristics for the genealogy}
\label{sect:heuristics}
In this section, we derive the  guiding heuristics for the phase transition between the semi- and fully-pulled regimes. Since the genealogy is given by an exchangeable coalescent directed by $\eta^N$ (Proposition \ref{prop:ExpModelMultCoa}), the phase transition will emerge from the asymptotics of the mass partition $\eta^{N}$ introduced in Definition \ref{definition-wi}.

\subsection{ Fully-pulled regime.} Let $\beta<\beta_c(\rho)$ and recall that $\alpha\equiv\alpha(\beta,\rho)$ in \eqref{eq:defAlpha} satisfies $\alpha>0$ in this case.
Set
$$
\forall y>0,  \ \ L^N(y) \ = \ \#\{ i  : N^\alpha  w_{\hat I_{i}^N} \geq y\},
$$
Let us first interpret this quantity. Recall that $w_i=({\cal E}_1+\cdots+{\cal E}_i)^{-1}$ and 
for $i>>1$, the law of large numbers imply that $w_i \approx \frac{1}{i}$. 
As a consequence, $L^N(y)$ can be interpreted as the number of elements in $(w_i)_{i=1}^{N^\rho}$
that are selected in the $\beta$-sampling (natural selection)
and whose rank is smaller than $N^\alpha/y$.

The guiding heuristic  is that 
sampling without replacement of the indices $I^N$ in Definition \ref{definition-wi} 
becomes equivalent in the limit as $N\to\infty$ to a sampling with replacement. This will be shown by a coupling argument in Proposition
\ref{cv:joint}. This implies that 
\begin{eqnarray*}
\P\left(L^N(y) = n\right) \ \approx \ {N \choose n} p^n(1-p)^{N-n}, \ \ \mbox{where $p=\frac{\sum_{N^\alpha  w_i \geq y}  w_i^\beta}{\sum_{i=1}^{N^\rho} w_i^\beta}$.} 
\end{eqnarray*}
Since $\beta<1$ and $\alpha >0$,  and $w_i\approx 1/i$ for large $i$, a Riemann approximation yields
\begin{eqnarray*}
p & \approx  \frac{\sum_{i\leq N^\alpha/y}\frac{1}{i^\beta}}{ \sum_{i=1}^{N^\rho} \frac{1}{i^\beta}} 
& \approx  \frac{1}{N} \left(\frac{1}{y^{1-\beta}}\right)^n
\end{eqnarray*}
so that, as $N\to\infty$ and for fixed $n\geq1$,
\begin{eqnarray*}
\P\left(L^N(y) = n\right) & \approx & \frac{1}{n!} e^{-\frac{1}{y^{1-\beta}}} (\frac{1}{y^{1-\beta}})^n.
\end{eqnarray*}
{Recall} $Z_1>\cdots Z_N>\cdots$ the atoms of a PPP($(1-\beta) \frac{dx}{x^{\beta-2}}$) arranged in decreasing order.
The previous computations show that 
$$
\P(L^N(y) = n) \ \approx \ \P\left( \#\{ Z_i : Z_i \geq y  \} =  n \right). 
$$
A generalization of the previous computation conducted in Lemma \ref{ref:lem22222} will show that 
$$
N^\alpha(w_{\hat I^N_1}, w_{\hat I^N_{2}},\cdots) \ \Longrightarrow (Z_1,Z_2,\cdots),
$$
so that the random mass partition can be approximated by
$$
\vec\eta^N \approx \vec\eta^{\infty}= \left(\frac{Z_i}{\sum_{j} Z_j}\right).
$$
The r.h.s. above, by  Proposition 6 in \cite{PitmanYor97}, has
Poisson-Dirichlet$(1-\beta,0)$ distribution.
{The above approximation will entail}  the convergence of {the $\vec\eta^N$-directed coalescent process} $\Pi_{n}^N$ to the Poisson-Dirichlet exchangeable coalescent as $N\to\infty$.

\bigskip

\subsection{ Semi-pulled regime.} Let us now consider $\beta_c(\rho)<\beta<1$ and recall that $\chi$ in \eqref{eq:defChi} satisfies $\chi\in (0,1]$ in this case.
For the sake of simplicity, let us again assume for the time being that {the sampling of the indices in Definition \ref{definition-wi}} occurs with replacement and define {$V^N(i)$}
as the average number of {times that index $i$ is drawn}. {Using again the law of large numbers to approximate $w_i\approx 1/i$}, we have

$$
{V^N(i) \ = \ N\E\[\frac{w_i^\beta}{\sum_{j=1}^{N^\rho} w_j^\beta}\]\ \approx \ N \frac{i^{-\beta}}{\sum_{j=1}^{N^\rho} j^{-\beta}}\sim i^{-\beta}N^{\beta\chi}.}
$$

Thus if
$$i<<N^\chi, \ \mbox{then} \ V^N(i)>>1, \ \ \mbox{and if} \ \   \ i>>N^\chi, \ \mbox{then} \ V^N(i)<<1.$$
In other words, all the entries with rank smaller than $N^\chi$ are very likely to be picked (Proposition \ref{prop:chooseFittest}), while an entry with rank higher than $N^\chi$ is likely to be missed.
This implies that 
\begin{eqnarray*}
\forall i \leq N^{\chi}, \  \eta^N_i & = & \frac{w_{\hat I_i^N}}{ \sum_{j=1}^N w_{\hat I_j^N} } 
 \approx  \frac{w_i}{\sum_{i=1}^{N^\chi} w_j + 
\sum_{j = N^{\chi}}^{N} w_{\hat I_j^N}}. 
\end{eqnarray*}
We will further show in Proposition \ref{prop:smallMassTail} that
$$
\sum_{j = N^{\chi}}^{N} w_{\hat I_j^N} << \sum_{i=1}^{N^\chi} w_j
$$
so that 
$$
\forall i \leq N^{\chi}, \  \eta^N_i \approx \frac{w_i}{\sum_{j=1}^{N^\chi} w_j}.
$$

{This suggests that the  the natural selection step becomes equivalent to truncation selection with an ``effective population size" $N^\chi$ as $N\to\infty$.
The second part of Theorem \ref{th:coaConvergence} is then a direct consequence of \cite{BrunetDerrida2012} on truncation selection for an ``effective population size'' $N^\chi$, where the genealogy is proved to converge to the Bolthausen-Sznitman coalescent. This last part of the argument will be detailed in Section \ref{sec:convBSC}.
}

\section{Proofs for the fully-pulled regime.}\label{sec:proofsWeakSelection}
\label{sect:a_pos}

In this section we assume throughout that $\beta<\beta_c(\rho) = 1-\frac{1}{\rho}$, so that
\begin{equation}
 \alpha(\beta,\rho) \equiv \alpha  = \rho - \frac{1}{1-\beta} > 0.
\end{equation}

\subsection{Convergence of mass partitions.}\label{sec:weakConvMassPartitions}
Define ${\hat{\bm V}}^N$ the random vector  
constructed as $\hat{\bm W}^N$ (see Definition \ref{definition-wi}) except that sampling is
done {\it with} replacement.
As for $I^N$, we denote by $J^N= (J_1^N,\cdots, J_N^N)$ 
the vector of sampled indices with replacement arranged in their sampling order, 
and by $\hat J^N$ the permutation {of $J^N$ in decreasing order so that $w_{\hat J_1^N} > \cdots > \hat w_{\hat J_N^N}$ is in increasing order.}

In the following, we write $\set{J}^N$ to refer to the set $\{j\colon \exists i\geq1 \text{ such that }  j=J^N_i\}$; i.e.
we regard $J^N$ as a discrete set and remove repetitions. We define $\set{I}^N$ analogously. 
We can couple  the random mass partitions $\tilde{\bm\eta}^N$ and $\bm\eta^N$ by coupling $J^N$ and $I^N$ as follows. 
Let  $\# \set{J}^N$ be the cardinality of $\set{J}^N$.
Given $\set{J}^N$, sample $r^N:=N - \#\set{J}^N$ indices in the set $[N^\rho]\setminus \set{J}^N$ 
without replacement.
Let $\set{K}^N$ be the resulting set.
We have the following lemma {whose proof is left to the appendix}. 
\begin{lemma}\label{le:indexCoupling}
The 
set of indices  
$\set{J}^N \cup \set{K}^N$ is identical in law to $\set{I}^N$.
\end{lemma}

The main objective of this section is to prove the following result
in which we characterize the limiting distribution of $\bm \eta^N$,
via the above coupling with $\hat{\bm \eta}^N$ 
and the simpler characterization of  $\hat{\bm \eta}^N$ as $N\to\infty$. 
This,
in turn, will readily give Theorem \ref{th:coaConvergence} i).

\begin{proposition}\label{cv:joint}
Assume $\hat{\bm \eta}^N$ and $\bm\eta^N$ are coupled as previously exposed.
Then 
$$\(\tilde{\bm\eta}^N,\bm\eta^N\) \Longrightarrow (\bm\eta^\infty, \bm\eta^{\infty})$$
where each coordinate is endowed with the $\ell^1(\R_+) \equiv \ell^1$ topology (the $L^1$ topology on the set of non-negative real sequences), the joint convergence is meant in the product topology, and 
$\bm\eta^{\infty}$ has the Poisson-Dirichlet $(1-\beta,0)$ distribution.
\end{proposition}

We break up the proof of this proposition into two main intermediary results.  In the following Lemma \ref{ref:lem22222} 
and its Corollary \ref{lemma:VN}, we prove $N^\alpha{\hat{\bm V}}^N\Longrightarrow (Z_1,Z_2,\cdots)$ where $(Z_i)$ is a PPP.
Then, in  Lemma \ref{lemma:approx-coincide}, we prove that 
the first $\ell\in\N$ coordinates of $N^\alpha{\hat{\bm V}}^N$ and $N^\alpha{\hat{\bm W}}^N$ in fact coincide with high probability as $N\to\infty$,
so that also $N^\alpha{\hat{\bm W}}^N\Longrightarrow(Z_1,Z_2,\cdots)$.
The convergence of the normalized sequences $\tilde{\bm\eta}^N$ and $\bm\eta^N$ to $\bm\eta^\infty$ will then directly follow by the Continuous Mapping theorem \cite{Billingsley99}.

\begin{lemma}\label{ref:lem22222}
Let $(Z_j)_{j\geq1}$ be the atoms of a Poisson random measure on {$\R_+$} with intensity
$(1-\beta)x^{\beta-2}dx$ ranked in decreasing order. Then, for the scaled random vector
$N^\alpha\hat{\bm V}^N$ we have, as $N\to\infty$,
\begin{equation*}
\(N^\alpha w_{\hat{J}^N_1}, \dots, N^\alpha w_{\hat{J}^N_N}, 0, \dots\)
\Rightarrow 
\(Z_1, Z_2, \dots\).
\end{equation*}
where the convergence is meant in $\ell^1(\R_+)$.
\end{lemma}

\begin{proof}
We provide an argument along the same lines as Lemmas 20 and 21 in \cite{Schweinsberg2003}. It 
remains to prove the conditions of [Theorem 3.2 \cite{Billingsley99}] which in our setting can be written as
that
\begin{equation}\label{eq:finite:coordinate}
\(N^\alpha w_{\hat{J}^N_1}, \dots, N^\alpha w_{\hat{J}^N_\ell}\)
\Rightarrow 
\(Z_1, \dots, Z_\ell\)
\end{equation}
(convergence for fixed $\ell$), and the uniform distance bound in probability which can be written as
\begin{equation}\label{eq:preRemainder}
 \forall\eps>0,\lim_{\ell\to\infty}\limsup_{N\to\infty}
 \P\(N^\alpha \sum_{k=\ell+1}^N w_{\hat{J}^N_k}  \geq \eps \)= 0.
\end{equation} 

We will prove \eqref{eq:finite:coordinate} conditional on a realization of the sequence $(w_i)_i$ such that $k w_k\to 1$ (note that by Lemma \ref{lem:local-cv}, the set of such realizations has probability $1$).
For any collection of positive
real numbers $\infty=z_0>z_1\geq z_2 \geq \cdots \geq z_\ell$, for $1\leq i\leq \ell$, define 
\begin{eqnarray*}
{\bf L}^N_i :=  \{k\in  J^N \ : \  z_i\leq N^{\alpha} w_k < z_{i-1}\}, \ \ \ \ \mbox{and $L^N_i\coloneqq \# {\bf L}^N_i $.}
\end{eqnarray*}
Here we recall that $J^N$ refers to the indices sampled with replacement.
Let $n_1,\dots,n_\ell \in \N$. Let $\P(\cdot | (w_i)_i)$ be the (regular) probability measure conditioned on the sequence $(w_i)$.
We have the multinomial formula
\begin{align*}
\P\bigg(L^N_1=n_1,\dots,L^N_\ell=n_\ell\  | \ (w_i)_i\bigg)
=&\frac{\(N\)_{n_1+\dots+n_\ell}}{n_1!\dots n_\ell!}
  \frac{\prod_{i=1}^\ell \(\sum_{ k: N^\alpha w_{k}\in [z_i, z_{i-1})  } w_k^{\beta}\)^{n_i}}
       {\(\sum_{k=1}^{N^\rho}w_k^{\beta}\)^{n_1+\dots+n_\ell}}\times\\
  &\(1-\frac{\sum_{ k: N^\alpha w_{k}\in [z_i,z_{i-1}) } \ w_k^{\beta}}{\sum_{k=1}^{N^\rho}w_k^{\beta}}\)^{N-n_1-\dots-n_\ell}.
\end{align*}
Since $k w_k \to1$ as $k\to\infty$, this easily implies, as $N\to\infty$, 

\begin{align*}
\(N\)_{n_1+\dots+n_\ell}\sim N^{n_1+\dots+n_\ell},\\
\sum_{k=1}^{\ceil{N^\rho}}w_k^{\beta} \  \sim \ 
\frac{N^{\rho(1-\beta)}}{(1-\beta)},\\
\( \sum_{ k: N^\alpha w_{k}\in [z_{i},z_{i-1})   } w_k^{\beta} \)^{n_i}
\sim
\( \frac{\(N^\alpha/z_i\)^{(1-\beta)}}{1-\beta}-
  \frac{\(N^\alpha/z_{i-1}\)^{(1-\beta)}}{1-\beta}\)^{n_i}, 
\end{align*}
and 
\begin{align*}
&\lim_{N\to\infty}   \(1-\frac{\sum_{k: N^\alpha w_{k} \geq z_\ell} \ w_k^{\beta}}{\sum_{k=1}^{N^\rho}w_k^{\beta}}\)^{N-n_1-\dots-n_\ell}\\
&=\exp\lbr -\lim_{N\to\infty} (N-n_1-\dots-n_\ell)\frac{N^{\alpha(1-\beta)}z_\ell^{\beta-1}}{N^{\rho(1-\beta)}} \rbr
=e^{-z_\ell^{\beta-1}}.
\end{align*}
As a consequence,
\begin{align*}
&\P\bigg(L^N_1=n_1,\dots,L^N_\ell=n_\ell\  | \ (w_i)_i\bigg)
\\
&\sim \frac{N^{n_1+\dots+n_\ell}}{n_1!\dots n_\ell!}
  \frac{\prod_{i=1}^\ell \( \(N^\alpha/z_i\)^{1-\beta}-\(N^\alpha/z_{i-1}\)^{1-\beta}\)^{n_i}}{N^{\rho(1-\beta)(n_1+\dots+n_\ell)}}
  e^{-z_\ell^{\beta-1}}\\
&=\frac{N^{n_1+\dots+n_\ell}}{n_1!\dots n_\ell!}\frac{\prod_{i=1}^\ell \( z_i^{\beta-1}-\(z_{i-1}\)^{\beta-1}\)^{n_i}}{N^{n_1+\dots+n_\ell}}
  e^{-z_\ell^{\beta-1}}\\
&=\prod_{i=1}^\ell \frac{e^{-(z_i^{\beta-1}-z_{i-1}^{\beta-1})}\( z_i^{\beta-1}-\(z_{i-1}\)^{\beta-1}\)^{n_i}}{n_i!}\\
&=\P\(Z([z_{i}, z_{i-1}))=n_i \text{ for }1\leq i \leq \ell\).
\end{align*}
Thus we have the a.s. limit
\begin{align*}
\lim_{N\to\infty}\P\left(N^\alpha w_{\hat{J}^N_1}\geq z_1, \dots, N^\alpha w_{\hat{J}^N_\ell}\geq z_\ell \ | \ (w_i)_i \right)
&=\lim_{N\to\infty}\P\left( L^N_i \geq i \text{ for }1\leq i\leq \ell \ | \ (w_i)_i \right)\\
&=\P\(Z([z_{i}, z_{i-1}))\geq i \text{ for }1\leq i \leq \ell\) \\
&=\P\(Z_1\geq z_1,\dots, Z_\ell\geq z_\ell\),
\end{align*}
so that \eqref{eq:finite:coordinate} holds by Lemma \ref{lem:local-cv}.

\bigskip

It remains to prove (\ref{eq:preRemainder}). By Markov's inequality,  \eqref{eq:preRemainder} is an easy consequence of the following limit 
\begin{equation}\label{eq:remainder}
\lim_{\ell\to\infty}\limsup_{N\to\infty} \E\[ \(N^\alpha\sum^N_{k=\ell+1} w_{\hat{J}^N_k}\)\wedge 1 \]=0.
\end{equation}
To prove the latter, let $0<\delta<1$. Since $\sum_{k=1}^\infty Z_k<\infty$ a.s. we may choose $\ell$ large enough to ensure 
$\P\(Z_\ell>\delta/2\)<\delta/2$. Having fixed $\ell$, equation
\eqref{eq:finite:coordinate} allows us to
choose $N_1$ large enough to ensure 
$
\P\(N^\alpha w_{\hat{J}_\ell}>\delta\)< \delta
$
for every $N\geq N_1$; thus, using that the sequence $\(N^\alpha w_{\hat{J}^N_\ell}\)_{\ell}$ is
non-increasing in the second line below, we obtain that for $N\geq N_1$
\begin{align*}
\E\[\(\sum_{k=\ell}^N N^\alpha w_{\hat{J}^N_k}\)\wedge 1\]
&\leq \P\(N^\alpha w_{\hat{J}^N_\ell}>\delta\) + \E\[\sum_{k=\ell}^N N^\alpha w_{\hat{J}^N_k};N^\alpha w_{\hat{J}_\ell}\leq\delta\]  \\
&\leq \delta + \E\[\sum_{k=1}^N N^\alpha w_{\hat{J}^N_k} \Ind{N^\alpha w_{\hat{J}_k}\leq \delta}\].
\end{align*}
From the large deviation bounds of Corollary \ref{cor:ldd} (
setting $b=\alpha$ and $a=1$ therein) and Remark \ref{rem:lddReplacement} in the appendix,
 the second term above is of order $\bO(\delta^\beta)$ as $\delta\to0$
for  large enough $N$. Thus \eqref{eq:remainder} follows by first taking $N$ large enough and then $\delta\to0$.

\end{proof}

\begin{corollary}\label{lemma:VN}
As $N\to\infty$, $\hat{\bm \eta}^N\Rightarrow \bm \eta^\infty$, where the convergence is meant in the $\ell^1$ topology.
\end{corollary}

\begin{proof} 
By Proposition 6 in \cite{PitmanYor97}, 
 $\(\frac{Z_1}{\sum_{k=1}^\infty Z_k}, \frac{Z_2}{\sum_{k=1}^\infty Z_k},\dots\)$ has Poisson-Dirichlet$(1-\beta,0)$ distribution.
The result follows from the previous lemma, the fact that the mapping $\nu \to \frac{\nu}{\norm{\nu}}$ is continuous in $\ell^1(\R^+)\setminus \ze$, and the Continuous Mapping theorem \cite{Billingsley99}. 
\end{proof}

In order to prove Proposition 
\ref{cv:joint}, it remains to show that $\eta^N$ and $\tilde \eta^N$ coincide at the limit.
Recall that under our coupling the set of sampled indices ${\cal I}^N$ can be decomposed into two sets  ${\cal J}^N$ and ${\cal K}^N$, where ${\cal J}^N$ is the set obtained by sampling $N$ individual {\it with} replacement.
Lemma \ref{ref:lem22222}  entails 
that elements in ${\cal J}^N$
have ranks of order $N^{-\alpha}$. The next result states that element in ${\cal K}^N$  must have a much smaller order. In particular, this will imply that the leading order terms in $\hat W^N$ and $\hat V^N$ must coincide.

\begin{lemma}\label{lemma:approx-coincide}
 There exists  $\delta>0$ such that 
$$ \P\bigg(\#\lbr k\in {\cal K}^N \colon w_k\geq \frac{1}{N^{\alpha+\delta}}\rbr = 0 \bigg) \to 1 \ \ \ \mbox{as $N\to\infty$}. $$
\end{lemma}
\begin{proof}
{By Markov's inequality and dominated convergence it is sufficient to find $\delta>0$ such  that}
$$
\E\[\#\lbr k\in K^N  \colon w_k \geq \frac{1}{N^{\alpha+\delta}}\rbr \wedge 1\ | \ (w_i)_i\] \to 0   \ \ \mbox{a.s. when $N\to\infty$.}
$$

Recall that $J^N$ denotes the indices sampled with replacement (with possible repetition), and that $\# J^N$ is the number of distinct
indices in $J^N$. Given $J^N$, the probability that $k\notin J^N$
is in ${\cal K}^N$ is always bounded from above by
$\frac{w^\beta_k}{\sum_{j=N}^{ [N^\rho]} w_j^\beta}$.
As a consequence, 
\begin{eqnarray*}
\E\[\#\lbr k\in K^N \colon w_k\geq \frac{1}{N^{\alpha+\delta}}   \rbr \   \ \vert \ (w_i)_i \] 
& \leq &  (N - \# J^N)  \sum_{k : w_k \geq \frac{1}{N^{\delta+\alpha}}} \frac{w^\beta_{k}}{{\sum_{j=N}^{N^\rho} w_j^\beta  }}.
\end{eqnarray*}
Integrating over $J^N$, and using the fact that $k w_k \to 1$ a.s., we have
\begin{align*}
&\limsup_{N\to\infty} \E\[\#\lbr k\in K^N \colon w_k \geq \frac{1}{N^{\alpha+\delta}}  \rbr \ | \ (w_i)_i  \] \\
& \leq   \limsup_{N\to\infty} \frac{1}{ \sum_{j= N}^{N^\rho} w_j^\beta}  \E\bigg[N - \# J^N \ | \ (w_i)_i  \bigg]  \  \sum_{k : w_k \geq \frac{1}{N^{\delta+\alpha}}}  w^\beta_{k} \\
&  \leq  \limsup_{N\to\infty}  \E\bigg[ N - \# J^N \ | \ (w_i)_i\bigg] \frac{N^{(\delta+\alpha)(1-\beta)}}{N^{\rho(1-\beta)} }.
\end{align*}
In order to evaluate the expected value on the r.h.s., we will use the simple bound
$$ N - \# J^N  \ \leq  \ \sum_{(i,j) \in[N] : i<j} 1_{J_i^N = J_j^N}. $$
We use again the fact that the probability that $w_k$
is sampled in $J^N$ at any step of the procedure is  bounded from above by
$\frac{w^\beta_k}{\sum_{j=N}^{[N^\rho]} w_j^\beta}$.
As a consequence,
\begin{eqnarray*}
\E\bigg[ N - \# J^N \ | \ (w_i)_i \bigg] & \leq &  \sum_{(i,j) \in[N] : i<j} \P({J_i^N = J_j^N} \ | \ (w_i)_i) 
 =   \sum_{(i,j) \in[N] : i<j} \sum_{k=1}^{\ceil{N^\rho}} \P\left({J_i^N = J_j^N = k} \ | \ (w_i)_i\right) \\
& \leq &  N^2 \sum_{k=1}^{\ceil{N^\rho}} \frac{w_k^{2\beta}}{ \left(\sum_{i=N}^{N^\rho} w_i^\beta\right)^2}.  \\
\end{eqnarray*}
Combining this with the previous estimates and using again the fact that $k w_k \to 1$, there exists a deterministic constant $c$ such that 
\begin{align*}
&\limsup_{N\to\infty}  \E\[\#\lbr k\in K^N \colon w_k \geq \frac{1}{N^{\alpha+\delta}}  \rbr  \  | \ (w_i)_i \]\\ 
&  \leq  c\limsup_{N\to\infty}  N^2 \(\frac{N^{\rho(1-2\beta)_+}\vee \log(N)}{N^{2\rho(1-\beta)}}\)\(\frac{N^{(\alpha+\delta)(1-\beta)}}{N^{\rho(1-\beta)}}\) \\
& \leq  c\limsup_{N\to\infty}  N \(\frac{N^{\rho(1-\beta)}}{N^{2\rho(1-\beta)}}\) N^{\delta(1-\beta)} 
 =   c \limsup_{N\to\infty} N^{1-\rho(1-\beta)+\delta(1-\beta)}.
\end{align*}

In the fully-pulled regime recall that $1-\rho(1-\beta)<0$. 
Thus, we can choose $\delta>0$ small enough such that the exponent on the r.h.s. remains negative. This completes the proof of the Lemma.
\end{proof}

\begin{proof}[Proof of Proposition \ref{cv:joint}]
We first prove the convergence of the scaled weights $N^\alpha\hat {\bf W}^N$ to $(Z_1, Z_2,\dots)$ in $\ell^1$.
As in Lemma \ref{ref:lem22222}, it is enough to prove 
\begin{equation}
\(N^\alpha w_{\hat{I}^N_1}, \dots, N^\alpha w_{\hat{I}^N_\ell}\)
\Rightarrow 
\(Z_1, \dots, Z_\ell\),
\end{equation}
for fixed $\ell$, and also the analogue of \eqref{eq:preRemainder}. 

Let $\eps>0$.
From Lemma \ref{lemma:approx-coincide}, under our coupling and as $N\to\infty$, we have
\begin{align}\label{eq:symm-diff}
&\P\bigg( \forall i \ \mbox{s.t.} \  N^\alpha w_{\hat J^N_i} > \eps \ : \ w_{\hat J^N_i} = w_{\hat I^N_i}   \bigg)\nonumber\\
&\geq \P\bigg(\#\lbr k\in {\cal K}^N \colon w_k\geq \frac{1}{N^{\alpha}\eps}\rbr = 0\bigg)
\to 1 \ \ \mbox{as $N\to\infty$}
\end{align}
so that
the sampled atoms $(w_{I^N_i})$ and $(w_{J^N_i})$ above level $N^{-\alpha} \eps$  coincide with a probability going to $1$ as $N\to\infty$ (see discussion preceding Lemma \ref{lemma:approx-coincide}).

From  Lemma \ref{ref:lem22222}, this implies that for every $z_{1},\cdots, z_\ell>0$, we have
$$
\P\(N^\alpha w_{\hat{I}^N_1}>z_1, \dots, N^\alpha w_{\hat{I}^N_\ell} >z_\ell \)
\Rightarrow
\P\(Z_1 > z_1,\dots,Z_\ell >z_\ell\).
$$
For the analogue of \eqref{eq:preRemainder} we let the reader convince himself 
that it can be proved using the same argument.   
Finally, from (\ref{eq:symm-diff})  we obtain that
the normalized sequences 
$(\tilde {\bm \eta}^N, \bm\eta^N)$ jointly converge to $(\bm\eta^\infty, \bm\eta^\infty)$.
\end{proof}

\begin{proof}[Proof of Theorem \ref{th:coaConvergence} i).]
Recall that $\Pi_n^N$ is an exchangeable coalescent conducted by the mass partition $\eta^N$. Proposition \eqref{cv:joint} provides the convergence of the mass partition
\begin{equation}\label{eq:PDconvMassPartitions}
{\bm \eta}^N \Rightarrow  {\bm \eta}^\infty
\end{equation}  in the $\ell^1$ topology. The convergence of $\Pi_n^N$ to the $n$-Poisson-Dirichlet coalescent follows by a direct application of [Theorem 2.1 i) \cite{MohleSagitov2001}]. 

(More precisely, let $\pi^N$ (resp. $\pi^\infty$) be the exchangeable partition of $\N$ generated by Kingman's paintbox procedure on the mass partition $\bm\eta^N$ (resp. $\bm\eta^\infty$). An application of 
Proposition 2.9  \cite{Bertoin2006} and \eqref{eq:PDconvMassPartitions} then give
 $\pi^N\Rightarrow \pi^\infty$. This entails condition (8) in \cite{MohleSagitov2001} which allows to apply  [Theorem 2.1 i) \cite{MohleSagitov2001}].)

We now prove convergence of the coalescence time  $\E\[T_2^N\]$ for a sample of size $2$. 
$T_2^N$ is a geometric r.v. (counting the number of trials) with parameter
\begin{equation}
c_N = \E\left[\sum_{i=1}^N (\eta_i^N)^2\right]. 
\end{equation}
Since the function $({\bf x}_1, {\bf x}_2, \dots) \to \sum_{i=1}^\infty \({\bf x}_i\)^2$ is continuous
and 
bounded in the unit ball of $\ell^1(\R^+)$, by Proposition \ref{cv:joint}, we get
\begin{equation}\label{ET2}
\lim_{N\to\infty} \E\[T_2^N\] = \lim_{N\to\infty}c_N^{-1}=\E\[\sum_{i=1}^\infty \(\eta_i^\infty\)^2\]^{-1}.
\end{equation}
By equation (6) in \cite{PitmanYor97},
\begin{align*}
\E\[\sum_{i=1}^\infty \(\eta_i^\infty\)^2\] &= \frac{1}{B(\beta,1-\beta)}\int_0^1 u^{1+\beta-1} (1-u)^{1-\beta-1} du
=\frac{B(1+\beta,1-\beta)}{B(\beta,1-\beta)}
=\beta.
\end{align*}
Similarly, 
\begin{align}\label{eq:RMRCApmf}
\P\(R^N_2 = i\) 
                = \frac{\E\[\(\eta_i^N\)^2\]}{c_N}.
\end{align}
Taking the limit as $N\to\infty$ as before, we obtain
$$
\lim_{N\to\infty} \P\(R^N_2 = i\) = \E\[\(\eta_i^\infty\)^2\]/\beta.
$$

\end{proof}

Finally, we prove the convergence in distribution of $\tilde R_1^N$ in the weak selection regime.
\begin{proof}[Proof of Theorem \ref{thm:typical} i)]
Observe that $\P\(\tilde R_1^N(t)=k\)=\E\[\eta^N_k\]$ which by \eqref{eq:PDconvMassPartitions}  converges to $\E\[\eta^\infty_k\]$.
\end{proof}

\subsection{Speed of evolution.}
\begin{proof}
[Proof of Theorem \ref{th:speedSelection} i)]
Recall from \eqref{eq:speedSelection} that
\begin{equation*}
\nu^N=\E\[\log\(\sum_{k=1}^N w_{I_k}\)\] = -\alpha \log(N) + \E\[\log\(\sum_{k=1}^N N^\alpha w_{I^N_k}\)\].
\end{equation*}
By Proposition \ref{cv:joint}, 
$$
\log\(\sum_{k=1}^N N^\alpha w_{I^N_k}\) \Longrightarrow \ \log\(\sum_{k=1}^\infty Z_k\)
$$
where 
$
{\mathcal S}_1:=\sum_{k=1}^\infty Z_k
$

has Laplace transform
$\E\[e^{-s{\mathcal S}_1}\]=\exp\lbr-\rho\beta s^{1-\beta}\rbr$ (see (12) in \cite{PitmanYor97}).
As a consequence, the proof is complete by a direct application 
of Proposition \ref{lem:log-cond} in the appendix where we show the uniform integrability of the collection of random variables $\left(\log\(\sum_{k=1}^N N^\alpha w_{I^N_k}\)\right)_{N}$.
\end{proof}


\section{Proofs for the semi-pulled regime.}\label{sec:proofsStrongSelectino}


In this section we assume throughout that $\beta>\beta_c(\rho) = 1-\frac{1}{\rho}$, and recall that then
\begin{equation*}
\chi\equiv \chi(\beta,\rho)=\frac{1-\rho(1-\beta)}{\beta}\in(0,1].
\end{equation*}

\subsection{Selection of the fittest individuals.}\label{sec:strongRegimeFittest}

In this section we prove that all the high fitness children positioned above level $\approx N^{-\chi}$ 
are all chosen during the selection step with probability converging to 1 as $N\to\infty$ 
(Proposition \ref{prop:chooseFittest}); whereas the ``aggregated''  fitness 
of individuals below level $\approx N^{-\chi}$ becomes negligible as $N\to\infty$
(Proposition \ref{prop:smallMassTail}). We will use these two heuristics in Section \ref{sec:convBSC} 
in order to ``reduce'' the dynamics of our model
to those of the original exponential model of Brunet, Derrida et al \cite{BrunetDerridaMullerMunier2006,BrunetDerridaMullerMunier2007,BrunetDerrida2012} with an effective population size equal to
$N^\chi$. \par

\begin{proposition}\label{prop:chooseFittest}
Let $0<\eps<\chi$. Define the event
$$ A_{N,\eps}\equiv A_N \ := \  \left\{ \lbr i \ : \ w_i \geq \frac{1}{N^{\chi-\eps}} \rbr \subseteq I^N  \right\},$$
then 
{$$\exists \eta>0\colon\quad \lim_{N\to\infty} N^\eta\P\left( A_{N,\eps}^c  \right) = 0.$$}

\end{proposition}
\begin{proof}

For every $i\in [N^\rho]$ \begin{align*}
\P\bigg(i \notin I^N \ \vert \ (w_i)_i\bigg) 
\ \leq \ \exp\lbr N \ln\(1- \frac{w_i^\beta}{\sum_{j=1}^{\lceil N^\rho\rceil} w_j^\beta} \) \rbr \ 
\ \leq \ \exp\lbr - \frac{N w_i^\beta}{ \sum_{j=1}^{\lceil N^\rho\rceil} w_j^\beta} \rbr.
\end{align*}

Summing over $i$ we obtain
\begin{align}\label{eq:fitIneq1}
\P\left( A_{N,\eps}^c  \  \vert \  (w_i)_i \right)  \leq \sum_{  i \ : \ w_i \geq \frac{1}{N^{\chi-\eps}}} \exp(-\frac{ N w_i^\beta}{ \sum_{j =1}^{\lceil N^\rho\rceil} w_j^\beta}).
\end{align}

Using the large deviation bound of Lemma \ref{lem:LLD}, for every $C>1-\beta$
we have
$$
{
\forall \eta>0,\quad
\P\left( \sum_{j=1}^{\lceil N^\rho\rceil}  w_j^\beta\leq  C^{-1} N^{\rho(1-\beta)} \right)=\lo{N^{-\eta}}.}
$$

Gathering the previous observations, we need to show that as $N\to\infty$, we have
$$
{\exists \eta>0\colon\quad N^\eta \ \E \left[   \sum_{ w_i \geq \frac{1}{N^{\chi-\eps}}  } \exp\left( - C\frac{ N w_i^{\beta}}{  N^{\rho(1-\beta)} } \right)    \right] \ \to \  0.}
$$
Indeed, as $N\to\infty$, 
\begin{align*}
 \E \left[ \sum_{  w_i  \geq \frac{1}{N^{\chi-\eps}}  }  \exp\left( -C \frac{ N w_i^{\beta}}{  N^{\rho(1-\beta)} } \right)    \right]
& = \int_{  \frac{1}{N^{\chi-\eps} }}^{\infty} \exp(-C  \frac{N x^{\beta}}{ N^{\rho(1-\beta)}}) \frac{dx}{x^2} \\
& =  \int_{  \frac{1}{N^{\chi-\eps} }}^{\infty} \exp(-C  N^{\beta \chi} x^{\beta}) \frac{dx}{x^2} \\
& =  N^{\chi} \int_{N^\eps}^\infty \exp(-C y^\beta) \frac{dy}{y^2} 
{=\bO(N^{\chi-\eps} e^{-CN^{\eps\beta}})}.
\end{align*}
\end{proof}

\begin{corollary}\label{cor:chooseFittest}
Let 
$$
\bar A_N  \coloneqq \lbr \sum_{k=1}^{N} w_{I^N_k} >\frac{\chi}{2}\log(N) \rbr
$$
then
$$
\lim_{N\to\infty} \log(N)\P(\bar A_N^c) = 0.
$$
\end{corollary}
\begin{proof}
By {Proposition} \ref{prop:chooseFittest} we may work on the set $A_N$ on which we have 
$\sum_{k=1}^{N} w_{I^N_k}\geq \sum_{w_i \geq \frac{1}{N^{\chi-\eps}}} w_i$. 
Let $a_N = \frac{\chi}{2}\log(N)$.
Then
by Campbell's formula combined with Chebyshev's inequality, 
\begin{align}\label{eq:sumetasbound2}
\lim_{N\to\infty}\log(N)\P\left(\bar A_N^c \right) & =  
\lim_{N\to\infty}\log(N)\P\left(  \sum_{i: w_i \geq \frac{1}{N^{\chi-\eps}}} w_i \leq
a_N \right) \nonumber \\
& \leq  \lim_{N\to\infty}
\log(N)  \exp(a_N) \E\bigg[\exp(- \sum_{i: w_i \geq \frac{1}{N^{\chi-\eps}}} w_i) \bigg] \nonumber\\
& =  \lim_{N\to\infty} \log(N) \exp(a_N) \exp\left(\int_{\frac{1}{N^{\chi-\eps}}}^\infty \frac{dx}{x^2}(\exp(-x)-1)\right) \nonumber\\
& =  \lim_{N\to\infty} \log(N) \exp(a_N) \exp\bigg( -\log(N^{\chi-\eps}) + \bO\(1\)  \bigg).
\end{align}

Take  $\eps$ small enough such that $\chi-\eps > \frac{\chi}{2}$, it is straightforward to check that the last term above converges to $0$ as $N\to\infty$. 
\end{proof}
We now prove that the ``aggregated'' fitness 
of individuals below level $\approx N^{-\chi}$ becomes negligible as $N\to\infty$.
\begin{proposition}\label{prop:smallMassTail}
Let $\eps\in(0,1)$. Define
$$
B_{N,\eps}:=\lbr \sum_{k=1}^N w_{I^N_k}\Ind{I^N_k > N^{(\chi + \eps)}} < N^{-\eps\beta/2} \rbr.
$$
Then
{
\begin{equation}\label{eq:probSmallTail}
 \lim_{N\to\infty} N^{-\eps\beta/2}  \P\(B_{N,\eps}^c \) = 0.
\end{equation}
}

\end{proposition}
\begin{proof}

By Lemma \ref{le:pppsIdentities} and
 standard  Cramer's estimates, the probability
of the event 
$\lbr w_{\ceil{N^{\chi+\eps}}} \geq \frac{(1+\eps)^{-1}}{\lceil N^{\chi+\eps}\rceil}\rbr={\{ \rexp_1+\cdots\rexp_{\lceil N^{\chi+\eps}\rceil}\leq (1+\eps)\lceil N^{\chi+\eps}\rceil\}}$
decays exponentially fast in $N$.
This implies
$$\E\bigg [\sum_{k=1}^N w_{I^N_k}\Ind{I^N_k>N^{(\chi + \eps)}}\bigg]
\leq \E\bigg[ \sum_{k=1}^N w_{I^N_k}\Ind{w_{I^N_k}< \frac{(1+\eps)^{-1}}{N^{(\chi + \eps)}}}\bigg]+\lo{N^{-\eta}}
$$
for every $\eta>0$.
Also, a direct application of Corollary \ref{cor:ldd} with $a=1$, $b=\chi+\eps$, and $\delta=(1+\eps)^{-1}$, gives
$$
\E\bigg[ \sum_{k=1}^N w_{I^N_k}\Ind{w_{I^N_k}< \frac{(1+\eps)^{-1}}{N^{(\chi + \eps)}}}\bigg]
\ = \ \bO\(\frac{1}{N^{\beta\eps}}\). 
$$
Then Proposition \ref{prop:smallMassTail} is completed by a direct application of Markov's inequality.
\end{proof}


\subsection{Convergence to the Bolthausen-Sznitman Coalescent.}\label{sec:convBSC}

{
Our next Proposition \ref{prop:mc}, which is a simple corollary to Proposition 3.1 in \cite{WencesJegousseJOMB2024},
gives a general criterion for the weak convergence to 
$\Lambda$-coalescents of discrete-time coalescent processes directed by a mass partition $(\eta_1^N,\cdots, \eta^N_N)$ . It is also a simplification of Lemma 7.5 in \cite{CortinesMallein2017} that is based on similar heuristics: equation \eqref{eq:nomultcoag2} below ensures that no simultaneous multiple mergers occur in the limit; 
whilst \eqref{eq:singlCoae1} ensures that the simple multiple mergers converge to those of the $\Lambda$-coalescent. 
We will  apply this criterion to prove convergence to the Bolthausen-Sznitman coalescent when the directing random mass partition has the form \eqref{def:mass-partition} and $\beta\in(\beta_c(\rho),1)$.
}

\begin{proposition}\label{prop:mc}
Let $(\eta_1^N,\cdots, \eta^N_N)$ be (a general) random mass partition.
Assume that  there exist $L_N, N\in\N,$ such that  

\begin{equation}\label{eq:nomultcoag2}
{\lim_{N\to\infty} L_N\sum_{i=2}^N 
\E\[\(\eta_{i}^N\)^{2}\] = 0.}
\end{equation}
Whereas for $b\geq2$,
\begin{equation}\label{eq:singlCoae1}
\lim_{N\to\infty} L_N\sum_{i=1}^N \E\[\(\eta_i^N\)^{b}\] = \int_0^1 x^{b-2}\Lambda(dx).
\end{equation}
Then $c_N \sim L_{N}^{-1}\Lambda([0,1])$ and the time-scaled process $(\Pi^N_n([L_Nt]); t\in \R_{+})$ 
converges in distribution on $D([0,\infty), \PS{n})$  to the $n-\Lambda$-coalescent.
\end{proposition}

The following Lemma \ref{coro:etabsto0} proves
 \eqref{eq:nomultcoag2} for our model; whereas 
 Lemmas \ref{lem:brunet-derrida} and \ref{lem:cv-poluy} prove \eqref{eq:singlCoae1} for $\Lambda(dx)=dx$. These,
together with Proposition \ref{prop:ExpModelMultCoa}, will imply the convergence of the
genealogy to the Bolthausen-Sznitman coalescent.

\begin{lemma}\label{coro:etabsto0}
Recall $\bm\eta^N$ defined as in \eqref{def:mass-partition}. Then, for every $b\geq 2$,
$$
\lim_{N\to\infty}\log(N)\E\[\sum_{i=2}^N \(\eta_i^N\)^b\] = 0.
$$
\end{lemma}
\begin{proof}
Since $\eta_i\leq w_i / \sum_{k=1}^N w_{I^N_k}$ for all $1\leq i \leq N$, by Corollary
\ref{cor:chooseFittest} we have, for every $0<\eps<{1}$, 
\begin{align*}
\E\[\sum_{i=2}^N \(\eta_i^N\)^b\] &\leq
\E\[\sum_{i=2}^N \(\eta_i^N\)^{2-\eps}\]\\ 
&\leq\E\[\sum_{i=2}^N \frac{w_i^{2-\eps}}{\(\chi\log(N)/2\)^{2-\eps}}\]+ \lo{\log(N)^{-1}}\\
&=\frac{\sum_{i=2}^N \E\[w_i^{2-\eps}\]}{\(\chi\log(N)/2\)^{2-\eps}} + \lo{\log(N)^{-1}}.
\end{align*} 
Recall that $1/w_i$ is identical in law with $\rexp_1+\dots +\rexp_i$
where $\(\rexp_i\)_{i\geq 1}$ are independent exponential r.v.s. Therefore, {using Stirling's approximation in the last asymptotic below,}
$$
\E\[w_i^{2-\eps}\] 
= \frac{\Gamma(i-(2-\eps))}{\Gamma(i)} \overset{i\to\infty}{\sim} i^{-(2-\eps)},  
$$  
and $\sum_{i=2}^N \E\[w_i^{2-\eps}\]<\infty$. This proves
$$
\E\[\sum_{i=2}^N \(\eta_i^N\)^b\] = \lo{\log(N)^{-1}}.
$$
\end{proof}

As mentioned at the beginning of section \ref{sec:strongRegimeFittest}, our proof heuristic rests on an approximation of our model by the original Brunet-Derrida model \cite{BrunetDerridaMullerMunier2006,BrunetDerridaMullerMunier2007,BrunetDerrida2012} with an ``effective'' population size equal to
$N^\chi$. 

The following Lemma \ref{lem:brunet-derrida} 
proves \eqref{eq:singlCoae1} in the Brunet-Derrida setting with
population size $N$, while also allowing for an additional error term $a_N$ 
in the normalization constant $\sum_{k=1}^{N} w_k$. The result for $a_N=0$ was shown in Eq. (29) in \cite{BrunetDerrida2012}; 
 we extend these asymptotics to the case $a_N\to0$. Then in Lemma \ref{lem:cv-poluy} we use this result
with $N=N^\chi$
to prove \eqref{eq:singlCoae1} for our model; i.e. for the 
weights $\(w_{I_i}\)_{1\leq i\leq N}$ and the corresponding
mass partition $\bm\eta^N$.

\begin{lemma}\label{lem:brunet-derrida}
Let $a_N\to 0$ a non-negative deterministic constant converging to $0$ as $N\to\infty$. Then as $N\to\infty$
\begin{equation*}
\lim_{N\to\infty} \log(N) \ \E\[\sum_{j=1}^{N} \(\frac{w_j}{a_N+ \sum_{k=1}^{N} w_k}\)^b\] 
= \frac{1}{b-1}. 
\end{equation*}
\end{lemma}
\begin{proof}

First observe that
$$
\E\[\sum_{j=1}^N\(\frac{w_j}{a_N + \sum_{k=1}^N w_k}\)^b\]
=\E\[\(\frac{1}{\frac{a_N}{\sum_{k=1}^N w_k}+1}\)^b\sum_{j=1}^N\(\frac{w_j}{ \sum_{k=1}^N w_k}\)^b\].
$$

Second, by Markov's inequality followed by Jensen's inequality, 
\begin{align*}
\P\(\frac{a_N}{\sum_{k=1}^N w_k}>\eps\)
&\leq \frac{1}{\eps}a_N \E\[\frac{1}{\sum_{k=1}^N w_k}\]
=  \frac{1}{\eps}\frac{a_N}{\sum_{k=1}^N k^{-1}} \E\[\frac{\sum_{k=1}^N k^{-1}}{\sum_{k=1}^N w_kkk^{-1}}\]\\
&\leq \lo{\log(N)^{-1}} \E\[ \frac{\sum_{k=1}^N w_k^{-1}k^{-2}}{\sum_{k=1}^N k^{-1}} \]
=\lo{\log(N)^{-1}} \frac{\sum_{k=1}^N kk^{-2}}{\sum_{k=1}^N k^{-1}},
\end{align*}
where $\E\[w_k^{-1}\]=k$. 
Thus, intersecting with the set $\lbr \frac{a_N}{\sum_{k=1}^N w_k}\leq\eps\rbr$
and using (29) in \cite{BrunetDerrida2012} (i.e. our lemma for $a_N=0$) we obtain, for every $\eps>0$, 
\begin{align*}
&\lim_{N\to\infty}\log(N)\E\[\(\frac{1}{\frac{a_N}{\sum_{k=1}^N w_k}+1}\)^b\sum_{j=1}^N\(\frac{w_j}{ \sum_{k=1}^N w_k}\)^b\]\\
&\geq  \lim_{N\to\infty} \frac{\log(N)}{\(\eps+1\)^b}\E\[\sum_{j=1}^N\(\frac{w_j}{ \sum_{k=1}^N w_k}\)^b\] + 
\lim_{N\to\infty} \log(N)\P\(\frac{a_N}{\sum_{k=1}^N w_k}>\eps\)\\
&= \(\frac{1}{\(\eps+1\)^b}\)\frac{1}{b-1}.
\end{align*}
We then obtain an inequality by taking $\eps\to0$. The reversed inequality follows from the fact that it holds for $a_N=0$ thanks to (29) in \cite{BrunetDerrida2012}.
\end{proof}

\begin{lemma}\label{lem:cv-poluy}
Let $b\geq2$. Then for $\bm\eta^N$ as in \eqref{def:mass-partition},
$$
\lim_{N\to\infty}\chi\log(N)\E\[\(\eta_1^N\)^b\] = \int_0^1 bx^{b-2}(1-x)dx = \frac{1}{b-1}.
$$
\end{lemma}
\begin{proof}
Let $\eps>0$. Let $A_{N,\eps}$ be as in Proposition \ref{prop:chooseFittest} and $B_{N,\eps}$ as in Proposition \ref{prop:smallMassTail}.
Recall that for every $\eps>0$,
\begin{equation}\label{eq:ref-con}
\lim_{N\to\infty} \ \log(N) \P( A^c_{N,\eps}), \  \log(N) \P( B^c_{N,\eps}) \ = \ 0.
\end{equation}
As a consequence, it is sufficient to prove that
\begin{equation}\label{eq:LNeta1-1}
\lim_{N\to\infty}\chi\log(N)\E\[\(\eta_1^N\)^b; A_{N,\eps}\cap B_{N,\eps}\] = \frac{1}{b-1}, \quad \forall b\geq 2.
\end{equation}
On the one hand observe that in the event $A_{N,\eps}\cap B_{N,\eps}\subset A_{N,\eps}$ we have
$$
\(\eta_1^N\)^b
\leq \sum_{j=1}^{\lceil N^{\chi-\eps}\rceil} \(\frac{w_j}{\sum_{k=1}^{\lceil N^{\chi-\eps}\rceil}w_k}\)^b
$$
so that taking expectations, Lemma \ref{lem:brunet-derrida} implies that
\begin{align*}
&\limsup_{N\to\infty}\chi\log(N)\E\[(\eta_1^N)^b; A_{N,\eps}\cap B_{N,\eps}\] \\
&\leq \limsup_{N\to\infty}\chi\log(N)\E\[\sum_{j=1}^{\lceil N^{\chi-\eps}\rceil}\(\frac{w_j}{\sum_{k=1}^{\lceil N^{\chi-\eps}\rceil}w_k}\)^b\]
=\frac{\chi}{\chi-\eps}\(\frac{1}{b-1}\).
\end{align*}

On the other hand note that by Lemma \ref{coro:etabsto0} and for $\eps>0$ such that $\chi+\eps < 1$ we have 
for $b\geq2$,
$$
\liminf_{N\to\infty}\chi\log(N)\E\[(\eta_1^N)^b\] = \liminf_{N\to\infty}\chi\log(N)
\E\[\sum_{k=1}^{\lceil N^{\chi+\eps}\rceil}(\eta_k^N)^b\].
$$
Intersecting again with the set $A_{N,\eps}\cap B_{N,\eps}\subset B_{N,\eps}$, we see that the r.h.s. above is greater than
\begin{align*}
\liminf_{N\to\infty}\chi\log(N)\E\[\sum_{j=1}^{\lceil N^{\chi+\eps}\rceil}\(\frac{w_j}{\sum_{k=1}^{\lceil N^{\chi+\eps}\rceil}w_k+N^{-\eps\beta/2}}\)^b;  A_{N,\eps}\cap B_{N,\eps}\].
\end{align*}
Since $\log(N)\P(A^c_{N,\eps}), \log(N)\P(B^c_{N,\eps})\to 0$ and the sum above is bounded by $1$, we thus obtain
\begin{align*}
&\liminf_{N\to\infty}\chi\log(N)\E\[\sum_{k=1}^{\lceil N^{\chi+\eps}\rceil}(\eta_k^N)^b;  A_{N,\eps}\cap B_{N,\eps}\] \\
&\geq \liminf_{N\to\infty}\chi\log(N)\E\[\sum_{j=1}^{\lceil N^{\chi+\eps}\rceil}\(\frac{w_j}{\sum_{k=1}^{\lceil N^{\chi+\eps}\rceil}w_k+N^{-\eps\beta/2}}\)^b\].
\end{align*}
By Lemma \ref{lem:brunet-derrida}, the r.h.s. above converges to $\frac{\chi}{\chi+\eps}\(\frac{1}{b-1}\)$.
The proof
is finished by taking $\eps\to0$ 
in the above two bounds for $\lim_{N\to\infty}\chi\log(N)\E\[(\eta_1^N)^b;  A_{N,\eps}\cap B_{N,\eps}\]$.
\end{proof}

\begin{proof}[Proof Theorem \ref{th:coaConvergence} ii).]
The convergence of the process $(\Pi^N_n([\chi \log(N) t]; t\geq0)$ and the
asymptotic for $\E\[T^N_2\]\sim c_N^{-1}$ (see (\ref{ET2})) 
follow directly from Proposition \ref{prop:mc} together with the  estimates in Lemmas \ref{coro:etabsto0} and \ref{lem:cv-poluy}. Finally, the 
asymptotic for $\E\[R^N_2\]$ follows from \eqref{eq:RMRCApmf} and Lemma \ref{coro:etabsto0}. 
\end{proof}


\subsection{Speed of evolution.}\label{sec:speedSemiPulled}

\begin{proof}[Proof Theorem \ref{th:speedSelection} ii)]
By Lemma \ref{lem:27} together with Propositions \ref{prop:chooseFittest} and \ref{prop:smallMassTail}  we have,
for every $\eps>0$,
$\E\[\log\(\sum_{k=1}^N w_{I_k}\)\] = \E\[\log\(\sum_{k=1}^N w_{I_k}\);A_{N,\eps},B_{N,\eps}\]+\lo{1}$ as $N\to\infty$,
which implies
\begin{align}\label{eq:BSvel-0}
&\lim_{N\to\infty}\abs{{\log(\chi\log N)} - \E\[\log\(\sum_{k=1}^N w_{I_k}\)\]}  \\
&=\lim_{N\to\infty} \abs{\log({\chi}\log N) - \E\[\log\(\sum_{k=1}^N w_{I_k}\);A_{N,\eps}\cap B_{N,\eps}\]}\nonumber
\end{align}
for every $\eps>0$.

Observe that on the event $A_{N,\eps}\cap B_{N,\eps}$  
we have
\begin{equation}\label{eq:BSvel-1}
\log\(\sum_{k=1}^{\lceil N^{\chi-\eps} \rceil} w_k\)
\leq \log\(\sum_{k=1}^N w_{I_k}\) 
\leq \log\(N^{-\eps\beta/2} + \sum_{k=1}^{\lceil N^{\chi+\eps} \rceil} w_k \).
\end{equation}

Also note that for any sequence $0\leq a_N\overset{N\to\infty}{\to}0$,
\begin{align}\label{eq:BSvel-2}
\E\[\log\(a_N + \sum_{k=1}^{N} w_k\)\]
&= \log\(\sum_{k=1}^N k^{-1}\) + \E\[\log\( \frac{a_N + \sum_{k=1}^{N} w_k}{\sum_{k=1}^{N} k^{-1}}\)\]
\end{align}
where we now prove that the second term on the r.h.s. converges to $0$ as $N\to\infty$. 
Indeed, on the one hand, by Jensen's inequality 
$$
\E\[\log\( \frac{a_N + \sum_{k=1}^{N} w_k}{\sum_{k=1}^{N} k^{-1}}\)\]
\leq \log\( \frac{a_N + \sum_{k=1}^{N} \E\[w_k\]}{\sum_{k=1}^{N} k^{-1}} \)\overset{N\to\infty}{\to}0
$$
where for the last limit we have used that, by Lemma \ref{le:pppsIdentities}, 

$$
\E\[w_k\] = \frac{\Gamma(k-1)}{\Gamma(k)} = \frac{1}{k-1}.
$$

On the other hand, by Jensen's inequality, and using that $a_N\geq0$,
\begin{align*}
&-\E\[\log\( \frac{a_N + \sum_{k=1}^{N} w_k}{\sum_{k=1}^{N} k^{-1}}\)\]
=\E\[\log\( \frac{\sum_{k=1}^{N} k^{-1}}{a_N + \sum_{k=1}^{N} w_k}\)\]\\
&\leq \log\(\E\[ \frac{\sum_{k=1}^{N} k^{-1}}{\sum_{k=1}^{N} w_k}\]\)
= \log\(\E\[\frac{\sum_{k=1}^{N} k^{-1}}{\sum_{k=1}^{N} w_kkk^{-1}}\]\)
\leq\log\(\E\[\frac{\sum_{k=1}^{N} w_k^{-1}k^{-2}}{\sum_{k=1}^{N} k^{-1}} \]\)\\
&=\log(1),
\end{align*}
where for the last equality we have again used Lemma \ref{le:pppsIdentities} which implies $\E\[w_k^{-1}\]=k$.

Thus, taking expectations in \eqref{eq:BSvel-1} and plugging in \eqref{eq:BSvel-2} we obtain, for every
$\eps>0$,
\begin{align*}
\log(\chi-\eps) + \log\log\(N\)+ \lo{1}
&\leq\E\[ \log\(\sum_{k=1}^N w_{I_k}\); A_{N,\eps}\cap B_{N,\eps}\]\\
&\leq \log(\chi+\eps) + \log\log\(N\)+ \lo{1}.
\end{align*}
This, together with \eqref{eq:BSvel-0}, implies
$$
\log(1-\frac{\eps}{\chi}) \leq \lim_{N\to\infty}\abs{\log({\chi}\log N)-\E\[ \log\(\sum_{k=1}^N w_{I_k}\)\]} \leq \log(1+\frac{\eps}{\chi}).
$$
The proof is finished by taking $\eps\to0$. 

\end{proof}



Finally, we prove the convergence in distribution of $R_1^N$ in the strong selection regime.
\begin{proof}[Proof of Theorem \ref{thm:typical} ii)]
Since the argument uses the same ingredients as before, we only provide a sketch of the proof and leave the details to the reader.
Let $[a,b] \subset (0,\chi)$. Then
\begin{equation}\label{eq:r1N1}
\P\left( \log({\tilde R_1^N}) / \log(N) \in[a,b]\right) \ = \ \E\left[ \frac{\sum_{k\in[N^a,N^b]} w_{\hat I_k} }{ \sum_{k \leq N} w_{\hat I_k}   }\right] 
\end{equation}
Recall that we take $b < \chi$. By Proposition \ref{prop:chooseFittest} and the fact that $k w_{k} \to 1$ a.s. as $k\to\infty$ it follows that 
$$
\P\left( \sum_{k\in[N^a,N^b]} w_{\hat I_k} = \sum_{k\in[N^a,N^b]} w_{k} \right) \ \to 1, \ \ \mbox{as $N\to\infty$.}
$$
On the one hand, 
$$
\lim_{N\to\infty}\frac{1}{\log(N)}\sum_{k\in[N^a,N^b]} w_{k} = \lim_{N\to\infty} \frac{1}{\log(N)} \sum_{k\in[N^a,N^b]} \frac{1}{k}  = (b-a) \ \ \ a.s..
$$
On the other hand, combing Propositions \ref{prop:chooseFittest} and \ref{prop:smallMassTail}, we have
$$
\lim_{N\to\infty} \frac{1}{\log(N)}\sum_{k \leq N} w_{\hat I_k}  \ =  \ \lim_{N\to\infty} \frac{1}{\log(N)}\sum_{k\leq N^{\chi}} w_k    =  \  \chi,
$$
where the limits hold in probability.
Using {\eqref{eq:r1N1}} and the bounded convergence theorem, this implies that
$$
\lim_{N\to\infty} \ \P\left( \log({\tilde R_1^N}) / \log(N) \in[a,b]\right)  \ \to \ (b-a)/\chi  
$$
which completes the proof of the Theorem.

\end{proof}

\appendix

\section{Large deviations for the natural selection step.}

In this section we provide large deviation estimates 
for sums of the $w_i$'s appearing in Definition \ref{definition-wi}. This will lead to estimates for the sum of powers of the
sampled $w_{I_i}$'s in the natural selection step. The latter 
{are} used in Sections \ref{sec:proofsWeakSelection} and \ref{sec:proofsStrongSelectino}
to study functionals of the vector $\bm\eta^N$, leading up to the main results exposed in the introduction.
\begin{lemma}\label{lem:local-cv}
Almost surely, $k w_k \to 1$ a.s as $k\to\infty$.
\end{lemma}
\begin{proof}
This follows by the strong law of large numbers for sums of
exponential r.vs.
\end{proof}

Given Lemma \ref{lem:local-cv}, a 
standard Riemann integral approximation argument yields for $a<1$ and as $b\to\infty$, that
$$
\sum_{i=1}^{b}w_i^{a} \approx \frac{1}{1-a}b^{1-a}
$$
with high probability.  In the next results, we provide large deviation estimates for the probability to deviate from this LLN. 
 
\begin{lemma}\label{lem:LLD}
Let $c > 1-\beta$ and $\delta\in[0,1)$. Define
$$
E_N \ = \ \left\{ \sum_{i= \ceil{N^\delta}}^{N} w_i^\beta \leq c^{-1} N^{1-\beta} \right\}.
$$
Then for every $\eta>0$, we have $\P(E_N) =o( N^{-\eta})$.
\end{lemma}
\begin{proof}
By Lemma \ref{le:pppsIdentities} and standard  Cramer large deviation estimates for sums of i.i.d exponential r.vs., $\P\(w_N\geq \frac{N^{-1}}{1-\eps}\)$ and
$\P\(w_{\ceil{N^\delta}}\leq \frac{N^{-\delta}}{1+\eps}\)$  are of order $\lo{N^{-\eta}}$
for any choice of $0<\eps<1$ {and $\eta\ge 0$}. Thus
\begin{align*}
\P\(E_N\)
&= \P\(E_N; w_N< \frac{N^{-1}}{1-\eps}, w_{\ceil{N^\delta}}> \frac{N^{-\delta}}{1+\eps}\) + \lo{N^{-\eta}}\\
&\leq  \P\(\int_{\frac{N^{-1}}{1-\eps}}^{\frac{N^{-\delta}}{1+\eps}} x^{\beta}\PPP(dx) \leq \frac{N^{1-\beta}}{c}\) + \lo{N^{-\eta}}.
\end{align*}
It only remains to prove that for some appropriate choice of $\eps>0$ we have
\begin{equation}\label{eq:LLD-0}
\P\(\int_{\frac{N^{-1}}{1-\eps}}^{\frac{N^{-\delta}}{1+\eps}} x^{\beta}\PPP(dx) 
\leq \frac{1}{c(1-\eps)^{1-\beta}}\(N(1-\eps)\)^{1-\beta}\) = \lo{N^{-\eta}}.
\end{equation}
Using Markov's inequality and Campbell's formula, for any $a_1<a_2$ and $b>0$ we have
\begin{align*}
\P\(\int_{a_1}^{a_2} x^{\beta}\PPP(dx) \leq b^{-1}a_1^{\beta-1}\)
&\leq e^{b^{-1}a_1^{\beta-1}}\E\[\exp\lbr-\int_{a_1}^{a_2} x^{\beta}\PPP(dx)\rbr\]\\
&= \exp\lbr b^{-1}a_1^{\beta-1} - \int_{a_1}^{a_2}\frac{1-e^{-x^\beta}}{x^2}dx \rbr.
\end{align*}
Next, for every $x>0$, 
$1-e^{-x^\beta} = x^\beta e^{-\theta}$ for some $\theta\in[0, x^\beta]$. Thus
$$
1-e^{-x^\beta}\geq e^{-a_2^\beta} x^\beta
$$
for all $x\in[a_1,a_2]$. As a consequence, 
\begin{align}\label{eq:LLD-1}
\log\(\P\(\int_{a_1}^{a_2} x^{\beta}\PPP(dx)
     < ba_1^{\beta-1}\)\)
&\leq b^{-1}a_1^{\beta-1} - e^{-a_2^\beta}\int_{a_1}^{a_2}x^{\beta-2}dx \\
&=a_1^{\beta-1}b^{-1}\(1 - \frac{b\(1-(a_2/a_1)^{\beta-1}\)}{1-\beta}e^{-a_2^\beta}\).\nonumber 
\end{align}
Set $a_1\equiv a_1^N= \frac{N^{-1}}{1-\eps}$, $a_2\equiv a_2^N=\frac{N^{-\delta}}{1+\eps}$ and 
$b=c(1-\eps)^{1-\beta}$ where we recall that $c$ is a constant larger than $1-\beta$. 
Take $0<\eps<1$ small enough to ensure that 
$\frac{b}{1-\beta}>1$, then, since {$a_2^N/a_1^N\to\infty$ and $\beta-1<0$}, we obtain
$
\(1 - \frac{b\(1-(a_2/a_1)^{\beta-1}\)}{1-\beta}e^{-a_2^\beta}\)<0
$
for $N$ large enough. Substituting in \eqref{eq:LLD-1} we get \eqref{eq:LLD-0}.
\end{proof}

\begin{lemma}\label{lem:LLD2}
Let $2c \in (0,1-\beta)$ and
$$
E_N \ = \ \left\{ \sum_{i=1}^{N} w_i^\beta \geq c^{-1} N^{(1-\beta)} \right\}.
$$
Then there exists $\eta>0$ such that $\P(E_n) = o(N^{-\eta} )$.
\end{lemma}
\begin{proof}
As in the previous Lemma, we have
$\P\(w_N< \frac{N^{-1}}{1+\eps}\)= o(N^{-\eta})$
for every $\eta>0$. Thus, by Lemma \ref{le:pppsIdentities} 
\begin{align*}
\P\(E_N\) 
&\leq \P\(\int_{\frac{N^{-1}}{1+\eps}}^\infty x^{-\beta}\PPP(dx) \geq \frac{N^{(1-\beta)}}{c}\)
      + \lo{N^{-\eta}}.
\end{align*}
It remains to prove that with an appropriate choice of $\eps>0$, the first term in the right hand side is of order 
$\lo{N^{-\eta}}$ for some $\eta>0$. 
 Using Markov's inequality and Campbell's formula as before, for any $a<b<\infty$ and $A>0$ we compute,
\begin{align*}
\P\(\int_{a}^{b} x^{\beta}\PPP(dx) \geq  A \)
&\leq \E\[\exp\lbr \int_{a}^{b} x^{\beta}\PPP(dx)\rbr \]e^{-A}\\
&= \exp\lbr -\int_{a}^{b}\frac{1-e^{x^\beta}}{x^2}dx - A \rbr.
\end{align*}
We have
$
1-e^{x^\beta}\geq -x^\beta e^{b^\beta}
$
for all $0<x<b$. Thus
\begin{align*}
\log\(\P\(\int_{a}^{b} x^{-\beta}\PPP(dx) \geq A \)\)
&\leq e^{b^\beta}\int_{a}^{b}x^{\beta-2} dx - A \\
&\leq \frac{a^{\beta-1}}{1-\beta}\(e^{b^\beta}- A\frac{(1-\beta)}{a^{\beta-1}}\).
\end{align*}
Take $a\equiv a_N=\frac{N^{-1}}{1+\eps}$ and $A \equiv A_N = \frac{N^{(1-\beta)}}{2c}$. By the previous 
 estimate and Markov's inequality
\begin{align*}
&\P\(\int_{\frac{N^{-1}}{1+\eps}}^\infty x^{-\beta}\PPP(dx) \geq \frac{N^{(1-\beta)}}{c}\)\\
&\leq \P\(\int_{b}^{\infty} x^{\beta} \PPP(dx) 
         \geq \frac{N^{(1-\beta)}}{2c}\) +
     \P\(\int_{a_N}^{b} x^{\beta} \PPP(dx)
         \geq  \frac{N^{(1-\beta)}}{2c} \) \\
&\leq 2ca_N^{1-\beta}\E\[\int_{b}^\infty x^\beta \PPP(dx)\] +
     \exp\lbr  \frac{a_N^{\beta-1}}{1-\beta}\(e^{b^\beta}-\frac{1-\beta}{2c(1+\eps)^{1-\beta}}\) \rbr.
\end{align*}
Since $\E\[\int_{b}^\infty x^\beta \PPP(dx)\]=\int_{b}^\infty x^{\beta-2}dx$ the first term above is of order 
$\bO\(a_N^{(1-\beta)}\)=\bO\(N^{\beta-1}\)$, whereas the second term vanishes exponentially fast whenever $b,\eps>0$ are chosen small enough to ensure
$
e^{b^\beta} - \frac{1-\beta}{2c(1+\eps)^{1-\beta}} <0.
$

\end{proof}

\begin{corollary}\label{cor:ldd}
Let $a>0,b\geq0$ and $\delta\geq0$. Assume that $A:=a+\beta-1>0$. 
Then
$$
\lim_{N\to\infty}N^{-\chi\beta+bA}\E\left[ \sum_{k=1}^N w_{I^N_{k}}^a \Ind{w_{I_k^N} \leq \frac{\delta}{N^b} } \right]
\leq (1-\beta)\frac{\delta^A}{A}.
$$
\end{corollary}
\begin{proof}
Using that 
\begin{equation}\label{eq:chosenProbBound}
\P\(I^N_i=k \ | \ (w_i)_i\)\leq \frac{w_k^\beta}{\sum_{j=N}^{N^\rho} w_j^\beta}, \quad \forall 1\leq i\leq N,1\leq k\leq \lceil N^\rho \rceil,
\end{equation} 
and thus, summing over $i$,
$\P\(k\in I^N \ | \ (w_i)_i\)\leq Nw_k^\beta/\sum_{j=N}^{N^\rho} w_j^\beta$ for every $k$, note that 
\begin{align*}
\E\bigg[\sum_{k=1}^N w^{a}_{I^N_k}\Ind{w_{I^N_k}< \frac{\delta}{N^b}} \ | \ (w_i)_i\bigg]
& \leq   N  \  \sum_{k=1}^{\ceil{N^\rho}}  \ w^a_{k}\Ind{w_{k}< \frac{\delta}{N^{b}}}  \ \frac{w_k^\beta}{\sum_{j=N}^{N^\rho} w_j^\beta}.
\end{align*}
We will also use the trivial bound
\begin{align}\label{eq:corlldTrivialBound}
\sum_{k=1}^N w^a_{I^N_k}\Ind{w_{I^N_k}< \frac{\delta}{N^{b}}}
& \leq   N  \  (\frac{\delta}{N^{b}})^a.
\end{align}
Let $\bar E_N= \left\{ \sum_{i=N}^{\lceil N^\rho \rceil} w_i^\beta \leq c^{-1} N^{\rho(1-\beta)} \right\}$, 
then by Lemma \ref{lem:LLD} (replace $N$ and $\delta$ therein by $\lceil N^\rho\rceil$ and $1/\rho$ respectively),
\begin{align*}
\E\bigg[ \sum_{k=1}^N w^a_{I^N_k}\Ind{w_{I^N_k}< \frac{1}{N^{b}}}\bigg]
& \leq    N  \ \E\[ \sum_{k=1}^{\ceil{N^\rho}}  \ w^a_{k}\Ind{w_{k}< \frac{\delta}{N^{b}}}  \ \frac{w_k^\beta}{\sum_{j=N}^{\ceil{N^\rho}} w_j^\beta} ; \bar E_N^c \] \ + \   N  \  \frac{\delta^a}{N^{ab}} \P(\bar E_N)\\
& \leq  c N^{1-\rho(1-\beta)}  \ \E\[ \sum_{k=1}^{\ceil{N^\rho}}  \ w^{a+\beta}_{k}\Ind{w_{k}< \frac{\delta}{N^{b}}}   \]   + N  \  \frac{\delta^a}{N^{ab}}  \P(\bar E_N) \\
& =   c N^{1-\rho(1-\beta)}   \int_0^{\frac{\delta}{N^b}}\  \frac{dx}{x^{2-a-\beta}}  +   N  \  \frac{\delta^a}{N^{ab}} \P(\bar E_N).
\end{align*}
where we used Lemma \ref{le:pppsIdentities} in the last line.
The result follows from Lemma \ref{lem:LLD} and subsequently taking $c\downarrow 1-\beta$.
\end{proof}

\begin{corollary}\label{cor:ldd2}
Let $a>0,b\geq0$. Assume that $A=a+\beta-1<0$. 
Then 
$$
\lim_{N\to\infty}N^{-\chi\beta + b  A}\E\left[ \sum_{k=1}^N w_{I^N_{k}}^a \Ind{w_{I_k^N} \geq \frac{1}{N^b} } \right] \ \leq \ (1-\beta) \frac{1}{\abs{A}} .
$$
\end{corollary}
\begin{proof}
This can be proved analogously to the previous result by Lemma \ref{lem:LLD}.
By the same argument as in the previous corollary, intersecting again with the set 
$\bar E_N = \left\{ \sum_{i=N}^{\lceil N^\rho \rceil} w_i^\beta \leq c^{-1} N^{\rho(1-\beta)} \right\}$, but now using the trivial bound
$$
\sum_{k=1}^N w_{I^N_{k}}^a \Ind{w_{I_k^N} \geq N^{-b} }  \leq 
Nw_1^a \Ind{w_{1} \geq \frac{1}{N^b} }
$$
instead of \eqref{eq:corlldTrivialBound}, we obtain
\begin{align*}
&\E\bigg[ \sum_{k=1}^N w^a_{I^N_k}\Ind{w_{I^N_k}> \frac{1}{N^{b}}}\bigg]\\
& \leq    N  \ \E\[ \sum_{k=1}^{\ceil{N^\rho}}  \ w^a_{k}\Ind{w_{k}> \frac{1}{N^{b}}}  \ \frac{w_k^\beta}{\sum_{j=N}^{\ceil{N^\rho}} w_j^\beta} , \bar E_N^c \] \ + \   N  \E\[w_1^a\Ind{w_1>N^{-b}}, \bar E_N\]\\
& \leq  c N^{1-\rho(1-\beta)}  \ \E\[ \sum_{k=1}^{\ceil{N^\rho}}  \ w^{a+\beta}_{k}\Ind{w_{k}> \frac{1}{N^{b}}}   \]   + N  \norm{w_1^a\Ind{w_1>N^{-b}}}_{\Lp{2}}\(\P\(\bar E_N\)\)^{1/2}.
\end{align*}
Using Lemma \ref{le:pppsIdentities} we have  
$\E\[w_1^{2a}\Ind{w_1>N^{-b}}\]=\int_{N^{-b}}^\infty x^{-2a}e^{-x}dx$ which is of polynomial order on $N$ for any choice of $a$.
Thus Lemma \ref{lem:LLD} acertains that the second term above is of order $\lo{N^{\chi\beta - b  A}}$. 

For the first term we compute, using {Lemma \ref{le:pppsIdentities} and} the hypothesis $A<0$,
\begin{align*}
\E\[ \sum_{k=1}^{\ceil{N^\rho}}  \ w^{a+\beta}_{k}\Ind{w_{k}> \frac{1}{N^{b}}}   \]
&= \int_{\frac{1}{N^b}}^\infty \frac{dx}{x^{2-\alpha+\beta}} 
= \frac{1}{\abs{A}} N^{-bA}
\end{align*}
{and the Corollary follows.}
\end{proof}

\begin{remark}\label{rem:lddReplacement}
Corollaries \ref{cor:ldd} and \ref{cor:ldd2} remain valid if we replace the sampled indices $I^N$ 
(sampled without replacement) by a sample $J^N$ with replacement from $[N^\rho]$ with the same weights 
$\(w_i^\beta\)_{i=1}^{\lceil N^\rho \rceil}$ as for $I^N$. Indeed, this follows by observing that
even in this case the
key bound in \eqref{eq:chosenProbBound} remains valid.
\end{remark}

\begin{lemma}\label{lem:27}
Let $(E_N)$ be a sequence of events such that, for some $\eta'>0$,
$$
\lim_{N\to\infty} N^{\eta'}\P(E_N) \ = \ 0.
$$
Then, for every $\eta>0$
$$
\lim_{N\to\infty} \E\bigg[ |\log(\sum_{k=1}^N w_{I^N_k} )|^\eta; E_N \bigg] \ = \ 0.
$$
\end{lemma}

\begin{proof}
Let $\eta>0$ be fixed. Since $Nw_{\lceil N^\rho\rceil}\leq \sum_{k=1}^N w_{I_k}\leq Nw_1$
we have
\begin{align}
&\E\[\abs{\log\(\sum_{k=1}^N w_{I_k} \)}^\eta;E_N\]\nonumber\\
&\leq  \E\[\log^\eta(Nw_1); \sum_{i=1}^N w_{I_{k}}>1, E_N\] + \E\[\log^\eta(w^{-1}_{\lceil N^\rho \rceil}/N); \sum_{i=1}^N w_{I_{k}}\leq 1 , E_N\] \nonumber \\
& \leq  \E\[\log^\eta(Nw_1); Nw_1>1, E_N\] + \E\[\log^\eta(w_{\lceil N^\rho \rceil}^{-1}/N); w^{-1}_{\lceil N^\rho \rceil}/N \geq 1, E_N\]  \label{eq:fino}
\end{align}
We first deal with the first term on the r.h.s. of the inequality.
Let $\eps>0$. There exists a constant $c_{\eps,\eta}$ such that for every $x\geq1$, we have $\log^\eta(x) \leq c_{\eps,\eta} w^\eps$.
By Cauchy-Schwarz inequality,
\begin{eqnarray*}
\nonumber
\E\[\log^\eta(Nw_1); Nw_1>1, E_N\]  & \leq & c_{\eps,\eta}  N^\eps \E\[ w_1^\eps; E_N \] \\
& \leq & c_{\eps,\eta} N^\eps \E\[w_1^{2\eps}\]^{1/2} \P(E_N)^{1/2}.
\end{eqnarray*}
Recall that $w_1$ is identical to the inverse of a standard exponential r.v. As a consequence, the latter expectation is finite for $\eps<1/2$. Further, picking $\eps$ small enough such that $\eps <2\eta'$, 
the assumption of our lemma implies that $N^\eps \P(E_N)^{1/2}\to0$ so that the first term on the r.h.s. of (\ref{eq:fino}) must go to $0$.

We turn to the second term.
By a similar argument as before, for any $\eps>0$, we have
$$
\E\[\log^\eta(w_{\lceil N^\rho \rceil}^{-1}/N); w^{-1}_{\lceil N^\rho \rceil}/N \geq 1, E_N\] 
\leq \frac{c_\eps}{N^\eps} \E\[\frac{1}{{w_{N^{\rho}}^{2\eps}}}\]^{1/2} \P(E_N)^{1/2}
$$
Lemma \ref{le:pppsIdentities} together with Stirling's approximation yield 
$\norm{w_{\lceil N^\rho \rceil}^{-\eps}}_{L^2}^2=\frac{\Gamma(\lceil N^\rho \rceil+2\eps)}{\Gamma(\lceil N^\rho \rceil)}\sim N^{\rho2\eps}$. Thus, taking $\eps$ such that $0<\eps(\rho-1)<2\eta'$, the second term in (\ref{eq:fino}) is also vanishing.
\end{proof}

\section{Uniform integrability}

\begin{proposition}\label{lem:log-cond}
The collection of r.vs. $\left\{\log\(\sum_{k=1}^N N^\alpha w_{I^N_k}\) \right\}_{N}$  is
uniformly integrable.
\end{proposition}

The following easy-to-derive technical lemma will be used in the upcoming proof of Proposition \ref{lem:log-cond}.
\begin{lemma}\label{lemmm}
Let $\eta>1$ and  $\delta>0$. Then there exists $c>0$ such that
$$
\forall b\geq 1, x\geq 0, \ \ \log^{\eta}(b+x) \leq  c b^\delta + c x.
$$
\end{lemma}

\begin{proof}[Proof of Proposition \ref{lem:log-cond}]
This boils down to proving that for some $\eta>1$, we have
\begin{eqnarray}
\limsup_{N\to\infty} \  \E\[\abs{\log\(N^\alpha\sum_{k=1}^N w_{I^N_k}\)}^{\eta}; N^\alpha\sum_{k=1}^N w_{I^N_k}<1\] <\infty \label{eq:limsup1}\\
\limsup_{N\to\infty} \E\[\abs{\log\(N^\alpha\sum_{k=1}^N w_{I^N_k}\)}^{\eta}; N^\alpha\sum_{k=1}^N w_{I^N_k}>1\]  <\infty \label{eq:limsup2}.
\end{eqnarray}
(See e.g. (3.18) in \cite{Billingsley99}.)
\par

\bigskip

{\bf Step 1.} We start by proving (\ref{eq:limsup1}).
Let $E_N$ be an arbitrary subset such that, for some $\eta'>0$,
\begin{equation}
\label{eq:small-set}
\lim_{N\to\infty} \ N^{\eta'} \P(E_N) \ = \ 0.
\end{equation}

From Lemma \ref{lem:27}, it is enough to show that for such $E_N$,
$$
\limsup_{N\to\infty} \ \E\[\abs{\log\(N^\alpha \sum_{k=1}^N w_{I_k}\)}^{\eta},  \sum_{k=1}^{N} \ N^\alpha w_{I_k} \leq 1, E_N^c \] \  < \ \infty.
$$
To ease the notation, we write $J^N \equiv J$ and $I^N\equiv I$. We also use the same notation for the ordered set $\hat I$ and $\hat J$. First, 
\begin{align}\label{eq:speedSelPD-0}
&\E\[\abs{\log\(N^\alpha \sum_{k=1}^N w_{I_k}\)}^{\eta};  \sum_{k=1}^{N} \ N^\alpha w_{I_k} \leq 1, E_N^c \] \nonumber\\
& \leq   \E\bigg[\(-\log\( N^\alpha  w_{\hat I_1}\)\)^{\eta}; { N^\alpha w_{\hat I_1} \leq 1, E_N^c   } \bigg] \nonumber\\
& =   \int_0^\infty \P\bigg(  - \log\(N^\alpha  w_{\hat I_1}\) \geq u^{\frac{1}{\eta}},  N^\alpha w_{\hat I_1} \leq 1,  E_N^c    \bigg) du \nonumber\\ 
& \leq   \int_0^\infty \P\bigg(   N^\alpha w_{\hat I_1} \leq \exp(-u^{\frac{1}{\eta}}), E_N^c \bigg) du \nonumber\\
&\leq  \int_0^\infty \P\bigg(   N^\alpha w_{\hat J_1} \leq \exp(-u^{\frac{1}{\eta}}), E_N^c  \bigg) du \nonumber\\
& =  \int_0^\infty \E\left[ \(1- \frac{\sum_{i : w_i \geq  \frac{1}{N^\alpha} \exp(-u^\frac{1}{\eta})  } w_i^\beta }{\sum_{i=1}^{N^\rho} w_i^\beta } \)^N; E_N^c\right] du \nonumber \\
& \leq  \int_0^\infty \E\left[ \exp\( -  N \frac{\sum_{i : w_i \geq  \frac{1}{N^\alpha} \exp(-u^\frac{1}{\eta})  } w_i^\beta }{ \sum_{i=1}^{N^\rho} w_i^\beta } \); E_N^c\right] du, \nonumber  \\
\end{align}
where for the third inequality we have used the stochastic inequality
$$
\max_{1\leq k \leq N}N^\alpha w_{J_k} \leq \max_{1\leq k \leq N}N^\alpha w_{I_k}
$$
(see  our coupling $(I^N,J^N) \equiv (I,J)$ described at the beginning of Section \ref{sect:a_pos}.)

\medskip

We now handle the expectation appearing in the r.h.s. of $\eqref{eq:speedSelPD-0}$. Let $2 c\in(0,1-\beta)$ and define
$$
E_N \ = \
\left\{ \sum_{i=1}^{N^\rho} w_i^\beta > c^{-1} N^{\rho(1-\beta)} \right\}.
$$
By Lemma \ref{lem:LLD2}, (\ref{eq:small-set}) is satisfied. Further, using (\ref{eq:speedSelPD-0}), we have
\begin{align*}
&\E\[\abs{\log\(N^\alpha \sum_{k=1}^N w_{I_k}\)}^{\eta};  \sum_{k=1}^{N} \ N^\alpha w_{I_k} \leq 1, E_N^c \]\\  
& \leq    \int_0^\infty \E\left[ \exp( - c N \frac{\sum_{i : w_i \geq  \frac{1}{N^\alpha} \exp(-u^\frac{1}{\eta})  } w_i^\beta }{N^{\rho(1-\beta)}} ) \right] du \\
& =   \int_0^\infty \E\left[ \exp( - c N^{\chi \beta} \sum_{i : w_i \geq  \frac{1}{N^\alpha} \exp(-u^\frac{1}{\eta})  } w_i^\beta  ) \right] du.
\end{align*}

Define $\bar c =  \min_{(0,1]} \frac{1}{x^\beta}(1-\exp(-c x^\beta))>0$.
The latter, together with Campbell's formula and Lemma \ref{le:pppsIdentities}, and the inequality
$N^{\chi-\alpha}\leq1$ in the last line, yield
\begin{align*}
 \E\left[ \exp( - c N^{\chi \beta} \sum_{i : w_i \geq  \frac{1}{N^\alpha} \exp(-u^\frac{1}{\eta})  } w_i^\beta  ) \right] & =   
 \exp\bigg( \int_{\frac{\exp(-u^{\frac{1}{\eta}})}{N^\alpha}}^\infty \left( e^{- c (N^{\chi }  x)^{\beta} }-1 \right) \frac{dx}{x^2} \bigg) \\
 & =  \exp\bigg( N^\chi \int_{ \frac{N^\chi}{N^\alpha}  \exp(-u^{\frac{1}{\eta}})}^\infty \left( e^{- c   x^{\beta} }-1 \right) \frac{dx}{x^2} \bigg) \\
 & \leq  \exp\bigg( N^\chi \int_{ \frac{N^\chi}{N^\alpha}  \exp(-u^{\frac{1}{\eta}})}^1 \left( e^{- c   x^{\beta} }-1 \right) \frac{dx}{x^2} \bigg) \\
 &  \leq  \exp\bigg( - \bar c N^\chi \int_{ \frac{N^\chi}{N^\alpha}   \exp(-u^{\frac{1}{\eta}})}^1 \frac{dx}{x^{2-\beta}} \bigg) \\
 & \leq  \exp(-\frac{ \bar c}{1-\beta} \exp( (1-\beta)u^{\frac{1}{\eta}} ) ).
\end{align*}
Finally, (\ref{eq:limsup1}) follows from the observation that 
$$
\int_0^\infty \exp\bigg( -\frac{\bar c}{1-\beta}\exp\((1-\beta)u^\frac{1}{\eta}\) \bigg) du<\infty.
$$

\bigskip

{\bf Step 2.} We now prove (\ref{eq:limsup2}).
Let $\delta>0$.  By Lemma \ref{lemmm}, we can pick $c>0$ to ensure that
$$
\forall b\geq 1, x\geq 0, \ \ \log^{\eta}(b+x) \leq  c b^\delta + c x.
$$

Then, on the set  $\{ N^\alpha\sum_{k=1}^N w_{I_k}>1 \}$, and using the fact that $\log(x) \leq x$ for $x>0$ and 
$\log^\eta(x)\leq \tilde c x$ for some $\tilde c>0$ and all $x>1$, we have
\begin{align*}
\log^{\eta}\(N^\alpha\sum_{k=1}^N w_{I_k}\) & =  \Ind{w_{\hat I_1}>N^{-\alpha} }   \log^{\eta}\(N^\alpha\sum_{k=1}^N w_{I_k} \Ind{w_{I_k}>N^{-\alpha}}
                   \ +\ N^\alpha\sum_{k=1}^N w_{I_k} \Ind{w_{I_k}\leq N^{-\alpha}}\) \\
                   &   + \Ind{ w_{\hat I_1}\leq N^{-\alpha}}  \log^{\eta}\(
                    N^\alpha\sum_{k=1}^N w_{I_k} \Ind{w_{I_k}\leq N^{-\alpha}}\)  \\
& \leq c  N^{\delta \alpha}(\sum_{k=1}^N w_{I_k}\Ind{w_{I_k}> N^{-\alpha}})^\delta 
     + (c+\tilde c) N^\alpha \sum_{k=1}^N w_{I_k} \Ind{w_{I_k}\leq N^{-\alpha}}. 
\end{align*}

First, a direct application of Corollary \ref{cor:ldd} with $a,\delta=1$ and $b=\alpha$, shows that 
\begin{align*}
N^\alpha \sum_{k=1}^N \E\big[w_{I_k}; w_{I_k}\leq  N^{-\alpha} \big] \ = \ \bO(1).
\end{align*}

Second, take $\delta<1-\beta<1$. By Minkowski's inequality with $q= \delta^{-1}$, we have 
$$
\(\sum_{k=1}^N w_{I_k}\Ind{w_{I_k}>N^{-\alpha}}\)^\delta
\leq \sum_{k=1}^N w^{\delta}_{I_k}\Ind{w_{I_k}>N^{-\alpha}}.
$$
Then, by a direct application of Corollary \ref{cor:ldd2} with $a=\delta$ and $b=\alpha$, 

\begin{align*}
\E\[ N^{\delta \alpha}\big(\sum_{k=1}^N w_{I_k}\Ind{w_{I_k}>N^{-\alpha}}\big)^\delta\]
&\leq N^{\delta\alpha} \E\[ \sum_{k=1}^N w^\delta_{I_k}\Ind{w_{I_k}>N^{-\alpha}} \] \\
& = \bO\(N^{\alpha \delta + 1 -\rho(1-\beta) + \alpha(1-\beta-\delta)}\) \ = \ \bO(1).
\end{align*}
\end{proof}


 \bibliographystyle{abbrv} 
 \bibliography{Math}

\end{document}